\documentclass[journal,draftcls,onecolumn,12pt,twoside]{IEEEtran}

  \usepackage[nocompress]{cite}
  \usepackage{times}
 \usepackage{amsmath}
\usepackage{graphicx}
 \usepackage{amssymb}
 \usepackage{subfigure}
 \usepackage{amsthm}
 \usepackage{algorithm}
 \usepackage[noend]{algorithmic}
 \usepackage{color}
      \usepackage{graphicx}
\usepackage{subfig}
\usepackage{multirow}
 \usepackage{url}
\usepackage{cite}
\allowdisplaybreaks

\newtheorem{theorem}{Theorem}
\newtheorem{proposition}[theorem]{Proposition}
\newtheorem{lemma}[theorem]{Lemma}

\normalsize

\ifCLASSINFOpdf
\else
\fi

\begin{document}
%
\title{A Unified Form of EVENODD and RDP  Codes and Their Efficient Decoding}
\author{Hanxu~Hou,~\IEEEmembership{Member,~IEEE,}
        Yunghsiang S. Han,~\IEEEmembership{Fellow,~IEEE,}
        ~Kenneth W. Shum,~\IEEEmembership{Senior~Member,~IEEE} and~Hui Li~\IEEEmembership{Member,~IEEE,}
\thanks{Hanxu Hou is with the School of Electrical Engineering \& Intelligentization, Dongguan University of Technology and  with the Shenzhen Key Lab of Information Theory \& Future Internet Architecture, Peking University Shenzhen Graduate School~(E-mail: houhanxu@163.com). Yunghsiang S. Han is with the School of Electrical Engineering \& Intelligentization, Dongguan University of Technology~(E-mail: yunghsiangh@gmail.com). Kenneth W. Shum is with the Institute of Network Coding, The Chinese University of Hong Kong~(E-mail: wkshum@inc.cuhk.edu.hk). Hui Li is with the Shenzhen Key Lab of Information Theory \& Future Internet Architecture, Future Network PKU Lab of National Major Research Infrastructure, Peking University Shenzhen Graduate School(E-mail: lih64@pkusz.edu.cn).}
 }

\markboth{IEEE Transactions on Communications}%
{Submitted paper}

\maketitle
\vspace{-1.8cm}
\begin{abstract}
 Array codes have been widely employed in storage systems, such as Redundant Arrays of Inexpensive Disks (RAID). The row-diagonal parity (RDP) codes and EVENODD codes are two popular double-parity array codes. As the capacity of hard disks increases,  better fault tolerance by using array codes with three or more parity disks is needed. Although many extensions of RDP codes and EVENODD codes have been proposed, the high decoding complexity is  the main drawback of them. In this paper, we present a new construction for all families of EVENODD  codes and RDP codes, and propose a unified form of  them. Under this unified form,  RDP codes can be treated as shortened codes of EVENODD codes. Moreover,  an efficient decoding algorithm based on an LU factorization of Vandermonde matrix is proposed when the number of continuous surviving parity columns is no less than the number of erased information columns. The new decoding algorithm is faster than the existing algorithms when more than three information columns fail. The proposed efficient decoding algorithm is also applicable to other Vandermonde array codes.
 Thus the proposed MDS array code is practically very meaningful for storage systems that need higher reliability.
\end{abstract}

\begin{IEEEkeywords}
RAID, array codes, EVENODD, RDP, efficient decoding, LU factorization.
\end{IEEEkeywords}

\IEEEpeerreviewmaketitle

\section{Introduction}
\IEEEPARstart{A}rray codes have been widely employed in storage systems, such as Redundant Arrays of Inexpensive Disks (RAID)~\cite{RAID89, RAID93}, for the purpose of enhancing data reliability. In the current RAID-6 system, two disks are dedicated to the storage of parity-check bits, so that any two disk failures can be tolerated. There are a lot of existing works on the design of array codes which can recover any two disks failures, such as the EVENODD codes~\cite{blaum1995evenodd} and the row-diagonal parity (RDP) codes~\cite{corbett2004row}.

As the capacities of hard disks are increasing in a much faster pace than the decreasing of bit error rates, the protection offered by double parities will soon be inadequate \cite{BeyondRAID}. The issue of reliability is more pronounced in solid-state drives, which have significant wear-out rates when the frequencies of disk writes are high. In order to tolerate three or more disk failures, the EVENODD codes were extended in \cite{blaum1996mds}, and the RDP codes were extended in \cite{RTP12,blaum2006family}. All of the above coding methods are binary array codes, whose codewords are $m\times n$ arrays with each entry belonging to the binary field $\mathbb{F}_2$, for some positive integers $m$ and $n$. Binary array codes enjoy the advantage that encoding and decoding can be done by Exclusive OR (XOR) operations. 
The $n$ disks are identified as $n$ columns, and the $m$ bits in each column are stored in the corresponding disk. A binary array code is said to be \emph{systematic} if, for some positive integer $r$ less than $n$, the right-most $r$ columns store the parity bits, while the left-most $k=n-r$ columns store the uncoded data bits. 
If the array code can tolerate arbitrary $r$ erasures, then it is called a maximum-distance separable (MDS) array code. In other words, in an MDS array code, the information bits can be recovered from any $k$ columns.

\subsection{Related Works}
There are many follow-up studies on EVENODD codes~\cite{blaum1995evenodd} and RDP codes~\cite{corbett2004row} along different directions, such as the extensions of fault tolerance \cite{blaum1996mds,blaum2002evenodd,RTP12}, the improvement of repair problem \cite{wang2010rebuilding,xiang2010optimal,RecoverRaid,Zhu2014On} and efficient decoding methods \cite{huang2008star,Jiang2013Improved,wang2012triple,Huang2016An} of their extensions. 

Huang and Xu~\cite{huang2008star} extended the EVENODD codes to be STAR codes with three parity columns. The EVENODD codes were extended by Blaum, Bruck and Vardy~\cite{blaum1996mds,blaum2002evenodd} for three or more parity columns, with the additional assumption that the multiplicative order of 2 mod $p$ is equal to $p-1$. A sufficient condition for the extended EVENODD codes to be MDS with more than eight parity columns is given in \cite{Hou2016On}. Goel and Corbett~\cite{RTP12} proposed the RTP codes that extend the RDP codes to tolerate three disk failures. Blaum \cite{blaum2006family} generalized the RDP codes that can correct more than three column erasures and showed that the extended EVENODD codes and generalized RDP codes share the same MDS property condition. Blaum and Roth \cite{BlaumRoth93} proposed Blaum-Roth codes, which are non-systematic MDS array codes  constructed over a Vandermonde matrix. Some efficient systematic encoding methods for Blaum-Roth codes are given in \cite{BlaumRoth93,Qian2011On,Guo2013An}. We call the existing MDS array codes in \cite{blaum1995evenodd,corbett2004row,blaum1996mds,blaum2002evenodd,RTP12,blaum2006family,BlaumRoth93,huang2008star,Jiang2013Improved,wang2012triple,Huang2016An} as Vandermonde MDS array codes, as their constructions are based on Vandermonde matrices.


 \emph{Decoding complexity} in this work is defined as the number of XORs required to recover the erased no more than $r$ columns (including information erasure and parity erasure) from surviving $k$ columns. There are many decoding methods for extended EVENODD codes \cite{Jiang2013Improved} and generalized RDP codes; however, most of them focus on $r=3$. Jiang \emph{et al.} \cite{Jiang2013Improved} proposed a decoding algorithm for extended EVENODD codes with $r=3$. To further reduce decoding complexity of the extended EVENODD codes with $r=3$, Huang and Xu~\cite{huang2008star} invented STAR codes. 
 One extension of RDP codes with three parity columns is RTP codes, whose decoding has been improved by Huang \emph{et al.}\cite{Huang2016An}. Two efficient interpolation-based encoding algorithms for Blaum-Roth codes were proposed in \cite{Qian2011On,Guo2013An}. However, the efficient algorithms in \cite{Qian2011On,Guo2013An} are not applicable to the decoding of the extended EVENODD codes and generalized RDP codes. An efficient erasure decoding method that solves Vandermonde linear system over a polynomial ring was given in \cite{BlaumRoth93} for Blaum-Roth codes, and the decoding method is also applicable to the erasure decoding of extended EVENODD codes if the number of information erasures is no larger than the number of continuous surviving parity columns.  There is no efficient decoding method for arbitrary erasures and one needs to employ the traditional decoding method such as Cramer's rule to recover the erased columns. 

\subsection{Contributions}
In this paper, we present a unified form  of EVENODD  codes and RDP  codes  that include the existing RDP codes and their extensions in \cite{corbett2004row,blaum2006family}, along with the existing EVENODD codes and their extensions in \cite{blaum1995evenodd,blaum1996mds,blaum2002evenodd}. Under this unified form,  these two families of codes are shown having a close relationship between each other.  Based on this unified form, we also propose a fast method for the recovery of failed columns. This method is based on a factorization of Vandermonde matrix into very sparse lower and upper triangular matrices. Similar to the decoding method in \cite{BlaumRoth93}, the proposed fast decoding method can recover up to $r$ erasures such that the number of information erasure is no larger than the number of continuous surviving parity columns. We then illustrate the methodology by applying it to EVENODD codes and RDP  codes. We compare the  decoding complexity of the proposed method  with those presented in \cite{BlaumRoth93} for the extended EVENODD codes and generalized RDP codes. The proposed method has lower decoding complexity than that of the decoding algorithm given in~\cite{BlaumRoth93}, and is also applicable to other Vandermonde MDS array codes.

\section{Unified Form of EVENODD Codes and RDP Codes}
\label{sec:array_code}
In this section, we first present  EVENODD codes and RDP  codes. Then, we give a unified form of them and illustrate that RDP codes are   shortened  EVENODD  codes under this form. 

The array codes considered in this paper contain $p-1$ rows and $k+r$ columns, where $p$ is an odd number.
In the following, we let $k$ and $r$ be positive integers which are both no larger than $p$. Let $\textbf{g}(\ell) =(g(0),g(1),\ldots,g(\ell-1))$ be an $\ell$-tuple consisting of $\ell$ distinct integers that range from 0 to $p-1$, where $\ell\leq p$. The $i$-th entry of column $j$ are denoted as $a_{i,j}$ and $b_{i,j}$ for EVENODD  codes and RDP codes respectively. The subscripts are taken modulo $p$ throughout the paper, if it is not specified.

\subsection{EVENODD  Codes}
For an odd $p\geq\{k,r\}$, we define the EVENODD code as follows. It is a $(p-1)\times (k+r)$ array code, with the first $k$ columns storing the information bits, and the last $r$ columns storing the parity bits.
For $j=0,1,\ldots,k-1$, column $j$ is called \emph{information column} that stores the information bits $a_{0,j},a_{1,j},\ldots,a_{p-2,j}$, and for $j=k,k+1,\ldots,k+r-1$, column $j$ is called \emph{parity column} that stores the parity bits $a_{0,j},a_{1,j},\ldots,a_{p-2,j}$. 

Given the $(p-1)\times k$ information array $[a_{i,j}]$ for $i=0,1,\ldots,p-2$ and $j=0,1,\ldots,k-1$, we add an extra \emph{imaginary row} $a_{p-1,j}=0$, for $j=0,1,\ldots, k-1$, to this information array. The parity bits in column $k$ are computed by
\begin{equation}
a_{i,k}=\sum_{j=0}^{k-1}a_{i,j} \text{ for }0\leq i \leq p-2,
\label{EVENODD:row}
\end{equation}
and the parity bits stored in column $k+\ell$, $\ell=1,2,\ldots,r-1$,
 are computed by  
\begin{equation}
a_{i,k+\ell}=a_{p-1,k+\ell}+\sum_{j=0}^{k-1}a_{i-\ell g(j),j} \text{ for }0\leq i \leq p-2,
\label{EVENODD:diag}
\end{equation}
where 
\begin{equation}
a_{p-1,k+\ell}=\sum_{j=0}^{k-1}a_{p-1-\ell g(j),j}.
\label{eq:skl}
\end{equation}
We denote the EVENODD  codes defined in the above equations as $\textsf{EVENODD}(p,k,r;\textbf{g}(k))$. The default values in $\textbf{g}(k)$ are $(0,1,\ldots,k-1)$, and we simply write $\textsf{EVENODD}(p,k,r)$ if the values in $\textbf{g}(k)$ are default. An example of $\textsf{EVENODD}(5,3,3;(0,1,4))$ is given in Table \ref{table:evenodd}. Under the above definition, the EVENODD code in \cite{blaum1995evenodd} is $\textsf{EVENODD}(p,p,2)$ with $\textbf{g}(k)=(0,1,\ldots,k-1)$, and the extended EVENODD code in \cite{blaum1996mds} is $\textsf{EVENODD}(p,p,r)$ with $\textbf{g}(k)=(0,1,\ldots,k-1)$.
\begin{table*}
\caption{Encoding of $\textsf{EVENODD}(5,3,3;(0,1,4))$. Note that, by \eqref{eq:skl}, $a_{4,4}=a_{3,1}+a_{0,2}$ and  $a_{4,5}=a_{2,1}+a_{1,2}.$}
\label{table:evenodd}
\begin{center}
\begin{tabular}{|c|c|c|c|c|c|} \hline
$a_{0,0}$ & $a_{0,1}$ &  $a_{0,2}$& $a_{0,3}=a_{0,0}+a_{0,1}+a_{0,2}$& $a_{0,4}=a_{0,0}+a_{1,2}+a_{4.4}$ & $a_{0,5}=a_{0,0}+a_{3,1}+a_{2,2}+a_{4,5}$\\ \hline 
$a_{1,0}$ & $a_{1,1}$ &  $a_{1,2}$& $a_{1,3}=a_{1,0}+a_{1,1}+a_{1,2}$& $a_{1,4}=a_{1,0}+a_{0,1}+a_{2,2}+a_{4,4}$ & $a_{1,5}=a_{1,0}+a_{3,2}+a_{4,5}$\\ \hline 
$a_{2,0}$ & $a_{2,1}$ &  $a_{2,2}$& $a_{2,3}=a_{2,0}+a_{2,1}+a_{2,2}$& $a_{2,4}=a_{2,0}+a_{1,1}+a_{3,2}+a_{4.4}$ & $a_{2,5}=a_{2,0}+a_{0,1}+a_{4,5}$\\ \hline 
$a_{3,0}$ & $a_{3,1}$ &  $a_{3,2}$& $a_{3,3}=a_{3,0}+a_{3,1}+a_{3,2}$& $a_{3,4}=a_{3,0}+a_{2,1}+a_{4.4}$ & $a_{3,5}=a_{3,0}+a_{1,1}+a_{0,2}+a_{4,5}$\\ \hline  
\end{tabular}
\end{center}
\vspace{-0.8cm}
\end{table*}

\subsection{RDP  Codes} 
RDP code is an array code of size $(p-1) \times (k+r)$. Given the parameters $k,r,p$ that satisfy $p \geq \max(k+1,r)$, we add an extra \emph{imaginary row} $b_{p-1,0}=b_{p-1,1}=\cdots=b_{p-1,k-1}=0$ to the $(p-1)\times k$ information array $[b_{i,j}]$, for $i=0,1,\ldots,p-2$ and $j=0,1,\ldots,k-1$, as in $\textsf{EVENODD}(p,k,r)$. The parity bits of the $\textsf{RDP}(p,k,r; \textbf{g}(k+1)) $ are computed as follows:
\begin{equation}
b_{i,k}=\sum_{j=0}^{k-1}b_{i,j} \text{ for }0\leq i \leq p-2,
\label{RDP:row}
\end{equation}
\begin{equation}
b_{i,k+\ell}=\sum_{j=0}^{k}b_{i-\ell g(j),j} \text{ for }0\leq i \leq p-2, 1\leq \ell \leq r-1.
\label{RDP:diag}
\end{equation}
Like $\textsf{EVENODD}(p,k,r)$, the default value of $\textbf{g}(k+1)$ are $(0,1,\ldots,k)$. 
The first 4 rows in Table \ref{table:rdp1} are the array of $\textsf{RDP}(5,3,3;(0,1,4,3))$. The RDP code in~\cite{corbett2004row} is $\textsf{RDP}(p,p-1,2)$ with $\textbf{g}(p)=(0,1,\ldots,p-1)$ and $\textsf{RDP}(p,p-1,r)$ is the extended RDP in \cite{blaum2006family}.



\subsection{Unified Form}
There is a close relationship between $\textsf{RDP}(p,k,r;\textbf{g}(k+1))$ and $\textsf{EVENODD}(p,k,r;\textbf{g}(k))$ when both array codes have the same number of parity columns. The relationship can be seen by augmenting the arrays as follows. For RDP  codes, we define the corresponding augmented array as a $p\times (k+r)$ array with the top $p-1$ rows the same as in $\textsf{RDP}(p,k,r)$, and the last row defined by $b_{p-1,j}=0$ for $0\leq j\leq k$ and
\begin{equation}
b_{p-1,k+\ell}=\sum_{j=0}^{k}b_{p-1-\ell g(j),j} \text{ for } 1\leq \ell \leq r-1.
\label{RDP:piggyback}
\end{equation}
Note that \eqref{RDP:piggyback} is the extension of \eqref{RDP:diag} when $i=p-1$. The auxiliary row in the augmented array is defined such that the column sums of columns $k+1,k+2,\ldots,k+r-1$ are equal to zero. The above claim is proved as follows.

\begin{lemma}
For $1\leq \ell \leq r-1$, we have $\sum_{i=0}^{p-1}b_{i,k+\ell}=0$.
\end{lemma}
\begin{proof}
The summation of all bits in column $k+\ell$ of the augmented array is the summation of all bits in columns 0 to $k$. Since the summation of all bits in column $k$ is the summation of all bits in columns 0 to $k-1$, we have that the summation of all bits in column $k+\ell$ is equal to 0.
\end{proof}

By the above lemma, we can compute $b_{p-1,k+\ell}$ for $\ell=1,2, \ldots, r-1$ as
\[
b_{p-1,k+\ell}=b_{0,k+\ell}+b_{1,k+\ell}+\cdots+b_{p-2,k+\ell}.
\]
An example of the augmented array code of $\textsf{RDP}(5,3,3;(0,1,4,3))$ is given in Table \ref{table:rdp1}.

\begin{table*}
\caption{The augmented array of $\textsf{RDP}(5,3,3;(0,1,4,3))$.}
\label{table:rdp1}
\begin{center}
\begin{tabular}{|c|c|c|c|c|c|} \hline
$b_{0,0}$ & $b_{0,1}$ &  $b_{0,2}$& $b_{0,3}=b_{0,0}+b_{0,1}+b_{0,2}$& $b_{0,4}=b_{0,0}+b_{1,2}+b_{2,3}$ & $b_{0,5}=b_{0,0}+b_{3,1}+b_{2,2}$\\ \hline 
$b_{1,0}$ & $b_{1,1}$ &  $b_{1,2}$& $b_{1,3}=b_{1,0}+b_{1,1}+b_{1,2}$& $b_{1,4}=b_{1,0}+b_{0,1}+b_{2,2}+b_{3,3}$ & $b_{1,5}=b_{1,0}+b_{3,2}+b_{0,3}$\\ \hline 
$b_{2,0}$ & $b_{2,1}$ &  $b_{2,2}$& $b_{2,3}=b_{2,0}+b_{2,1}+b_{2,2}$& $b_{2,4}=b_{2,0}+b_{1,1}+b_{3,2}$ & $b_{2,5}=b_{2,0}+b_{0,1}+b_{1,3}$\\ \hline 
$b_{3,0}$ & $b_{3,1}$ &  $b_{3,2}$& $b_{3,3}=b_{3,0}+b_{3,1}+b_{3,2}$& $b_{3,4}=b_{3,0}+b_{2,1}+b_{0,3}$ & $b_{3,5}=b_{3,0}+b_{1,1}+b_{0,2}+b_{2,3}$\\ \hline \hline
0 & 0 & 0 & 0 & $b_{4,4}=b_{3,1}+b_{0,2}+b_{1,3}$ & $b_{4,5}=b_{2,1}+b_{1,2}+b_{3,3}$\\ \hline
\end{tabular}
\end{center}
\vspace{-0.8cm}
\end{table*}

Similarly, for an $\textsf{EVENODD}(p,k,r;\textbf{g}(k))$, the augmented array is a $p\times (k+r)$ array $[a'_{i,j}]$ defined as follows. The first $k+1$ columns are the same as those of $\textsf{EVENODD}(p,k,r;\textbf{g}(k))$, i.e., for $j=0,1,\ldots,k$ and $i=0,1,\ldots,p-1$, $a'_{i,j}=a_{i,j}$. 
For $\ell=1,2,\ldots,r-1$, we define the parity bits in column $k+\ell$ as
\begin{equation}
\label{eq:evenodd1}
a'_{i,k+\ell}:=\sum_{j=0}^{k-1}a_{i-\ell g(j),j} \text{ for } 0\leq i\leq p-1.
\end{equation}
We note that $a'_{p-1,k+\ell}$ is the same as $a_{p-1,k+\ell}$ defined in \eqref{eq:skl}. According to \eqref{EVENODD:diag}, the parity bits in column $k+\ell$ of $\textsf{EVENODD}(p,k,r;\textbf{g}(k))$ can be obtained from the augmented array by 
\[
a_{i,k+\ell}=a'_{i,k+\ell}+a'_{p-1,k+\ell}.
\]
\begin{lemma}
The bits in column $k+\ell$ for $\ell=1,2,\ldots,r-1$ of the augmented array can be obtained from $\textsf{EVENODD}(p,k,r, g(k))$ by 
\begin{equation}
a'_{p-1,k+\ell}=\sum_{i=0}^{p-2}(a_{i,k}+a_{i,k+\ell}), \text{ and }
\label{eq:augmented1}
\end{equation}
\begin{equation}
a'_{i,k+\ell}=a_{i,k+\ell}+a'_{p-1,k+\ell} \text{ for } i=0,1,\ldots,p-2.
\label{eq:augmented2}
\end{equation}
\label{lm:evenodd}
\end{lemma}
\begin{proof}
Note  that 
\begin{eqnarray}
&&\sum_{i=0}^{p-2}(a_{i,k}+a_{i,k+\ell})
=\sum_{i=0}^{p-2}(\sum_{j=0}^{k-1}a_{i,j})+
\sum_{i=0}^{p-2}(a_{p-1,k+\ell}+\sum_{j=0}^{k-1}a_{i-\ell g(j),j})\label{eq:a-1}\\
&=&\sum_{i=0}^{p-2}a_{i,0}+\cdots+\sum_{i=0}^{p-2}a_{i,k-1}+\sum_{i=0}^{p-2}a_{i-\ell g(0),0}+\cdots+\sum_{i=0}^{p-2}a_{i-\ell g(k-1),k-1}+\underbrace{(a_{p-1,k+\ell}+\cdots+a_{p-1,k+\ell})}_{p-1}\nonumber\\
&=&\sum_{i=0}^{p-1}a_{i,0}+\cdots+\sum_{i=0}^{p-1}a_{i,k-1}+\sum_{i=0}^{p-2}a_{i-\ell g(0),0}+\cdots+\sum_{i=0}^{p-2}a_{i-\ell g(k-1),k-1}\label{eq:a-2}\\
&=&\sum_{i=0}^{p-1}a_{i-\ell g(0),0}+\cdots+\sum_{i=0}^{p-1}a_{i-\ell g(k-1),k-1}+\sum_{i=0}^{p-2}a_{i-\ell g(0),0}+\cdots+\sum_{i=0}^{p-2}a_{i-\ell g(k-1),k-1}\label{eq:a-3}\\
&=&(a_{p-1-\ell g(0),0}+a_{p-1-\ell g(1),1}+\cdots+a_{p-1-\ell g(k-1),k-1})\nonumber\\
&=&a'_{p-1,k+\ell},\nonumber
\end{eqnarray}
where \eqref{eq:a-1} comes from \eqref{EVENODD:row} and \eqref{EVENODD:diag}, \eqref{eq:a-2} comes from that $a_{p-1,j}=0$ for $j=0,1,\ldots,k-1$, and \eqref{eq:a-3} comes from the fact that 
\[
\{-\ell g(j), 1-\ell g(j),\ldots, p-1-\ell g(j) \} = \{0,1,\ldots,p-1 \} \bmod p
\]
for $1\leq \ell \leq r-1$, $0\leq g(j)\leq p-1$.
Therefore, we can obtain the bit $a'_{p-1,k+\ell}$ by \eqref{eq:augmented1} and the other bits in parity column $k+\ell$ by \eqref{eq:augmented2}.
\end{proof}

The augmented array of $\textsf{EVENODD}(5,3,3;(0,1,4))$ is given in Table \ref{table:evenodd1}.

\begin{table*}
\caption{The augmented array of $\textsf{EVENODD}(5,3,3;(0,1,4))$.}
\label{table:evenodd1}
\begin{center}
\begin{tabular}{|c|c|c|c|c|c|} \hline
$a_{0,0}$ & $a_{0,1}$ &  $a_{0,2}$& $a_{0,3}=a_{0,0}+a_{0,1}+a_{0,2}$& $a_{0,4}=a_{0,0}+a_{1,2}$ & $a_{0,5}=a_{0,0}+a_{3,1}+a_{2,2}$\\ \hline 
$a_{1,0}$ & $a_{1,1}$ &  $a_{1,2}$& $a_{1,3}=a_{1,0}+a_{1,1}+a_{1,2}$& $a_{1,4}=a_{1,0}+a_{0,1}+a_{2,2}$ & $a_{1,5}=a_{1,0}+a_{3,2}$\\ \hline 
$a_{2,0}$ & $a_{2,1}$ &  $a_{2,2}$& $a_{2,3}=a_{2,0}+a_{2,1}+a_{2,2}$& $a_{2,4}=a_{2,0}+a_{1,1}+a_{3,2}$ & $a_{2,5}=a_{2,0}+a_{0,1}$\\ \hline 
$a_{3,0}$ & $a_{3,1}$ &  $a_{3,2}$& $a_{3,3}=a_{3,0}+a_{3,1}+a_{3,2}$& $a_{3,4}=a_{3,0}+a_{2,1}$ & $a_{3,5}=a_{3,0}+a_{1,1}+a_{0,2}$\\ \hline  \hline
0 & 0 & 0 & 0 & $a_{4.4}=a_{3,1}+a_{0,2}$ & $a_{4,5}=a_{2,1}+a_{1,2}$ \\ \hline
\end{tabular}
\end{center}
\vspace{-0.8cm}
\end{table*}

The augmented array of $\textsf{RDP}(p,k,r;\textbf{g}(k+1))$ can be obtained from shortening the augmented array of $\textsf{EVENODD}(p,k+1,r;\textbf{g}(k+1))$ and we summarize this fact in the following.
\begin{proposition}
\label{prop:shorten}
Let $\textbf{g}(k+1)$ of $\textsf{RDP}(p,k,r;\textbf{g}(k+1))$ be the same as $\textbf{g}(k+1)$ of $\textsf{EVENODD}(p,k+1,r;\textbf{g}(k+1))$. The augmented array of $\textsf{RDP}(p,k,r;\textbf{g}(k+1))$  can be obtained from shortening the augmented array of $\textsf{EVENODD}(p,k+1,r;\textbf{g}(k+1))$ as follows: (i) imposing the following additional constraint on the information bits 
\begin{equation}
a'_{i,k}=a'_{i,0}+a'_{i,1}+\cdots+a'_{i,k-1}
\label{eq:add}
\end{equation}
for $i=0,1,\ldots,p-1$;
(ii)  removing column $k+1$ of the augmented array of $\textsf{EVENODD}(p,k+1,r;\textbf{g}(k+1))$.
\end{proposition}
\begin{proof}
Consider the augmented array of $\textsf{EVENODD}(p,k+1,r;\textbf{g}(k+1))$ and assume that the information bits of column $k$ satisfy~\eqref{eq:add}. By~\eqref{EVENODD:row}, the parity bits in column $k+1$ are all zeros. After deleting column $k+1$ from the augmented array of $\textsf{EVENODD}(p,k+1,r;\textbf{g}(k+1))$ and reindexing the columns after this deleted column by reducing all indices by one, we have a new array with $k+r$ columns of a shortened $\textsf{EVENODD}(p,k,r;\textbf{g}(k))$.
Let the augmented array of $\textsf{RDP}(p,k,r;\textbf{g}(k+1))$ with the $k$ information columns being the same as the first $k$ information columns of the augmented array of $\textsf{EVENODD}(p,k+1,r;\textbf{g}(k+1))$ such that these columns are the same as those of the array of the shortened $\textsf{EVENODD}(p,k,r;\textbf{g}(k))$.  Then column $k$ of the augmented array of $\textsf{RDP}(p,k,r;\textbf{g}(k+1))$ is the same as column $k$ of the array of the shortened $\textsf{EVENODD}(p,k,r;\textbf{g}(k))$ according to \eqref{eq:add} and \eqref{RDP:row}. Recall that the bit $b_{i,k+\ell}$ in column $k+\ell$,  $i=0,1,\ldots,p-1$ and $\ell=2,3,\ldots,r-1$, of the augmented array of $\textsf{RDP}(p,k,r;\textbf{g}(k+1))$ is computed by \eqref{RDP:diag} (or \eqref{RDP:piggyback}). Since $a'_{i,j}=a_{i,j}=b_{i,j}$ for $i=0,1,\ldots,p-1$ and $j=0,1,\ldots,k$, $b_{i,k+\ell}$ is the same as $a'_{i,k+\ell}$ in the array of the shortened $\textsf{EVENODD}(p,k,r;\textbf{g}(k))$ that is defined by \eqref{eq:evenodd1}. Therefore, we can obtain the augmented $\textsf{RDP}(p,k,r;\textbf{g}(k+1))$ by shortening the augmented $\textsf{EVENODD}(p,k+1,r;\textbf{g}(k+1))$ by imposing the condition \eqref{eq:add} and removing column $k+1$, and this completes the proof.
\end{proof}
By Proposition \ref{prop:shorten}, the unified form of $\textsf{RDP}(p,k,r;\textbf{g}(k+1))$  and $\textsf{EVENODD}(p,k,r;\textbf{g}(k))$ is the augmented array of  $\textsf{EVENODD}(p,k+1,r;\textbf{g}(k+1))$.
In the following, we focus on $\textsf{EVENODD}(p,k,r;\textbf{g}(k))$, as the augmented array of $\textsf{RDP}(p,k,r;\textbf{g}(k+1))$ can be viewed as the shorten augmented array of $\textsf{EVENODD}(p,k+1,r;\textbf{g}(k+1))$.

\section{Algebraic Representation} 
Let $\mathbb{F}_2[x]$ be the ring of polynomials over binary field $\mathbb{F}_2$, and $R_p$ be the quotient ring $\mathbb{F}_2[x]/(1+x^p)$. An element in $R_p$ can be represented by a polynomial of degree strictly less than $p$ with coefficients in $\mathbb{F}_2$, we will refer to an element of $R_p$ as a polynomial in the sequel. Note that the multiplication of two polynomials in $R_p$ is performed under modulo $1+x^p$.

The ring $R_p$ has been discussed in  \cite{shumregenerating,Hou2016BASIC} and has been used in designing regenerating codes with low computational complexity. Let
\[
M_p(x):=1+x+\cdots+x^{p-1}.
\]
$R_p$ is isomorphic to a direct sum of two finite fields $\mathbb{F}_2[x]/(1+x)$ and 
$\mathbb{F}_{2}[x]/M_p(x)$\footnote{When 2 is a primitive element in $\mathbb{F}_p$, $\mathbb{F}_{2}[x]/M_p(x)$ is a finite field.} if and only if 2 is a primitive element in $\mathbb{F}_p$~\cite{Fenn1997Bit}.   In~\cite{Silverman1999Fast}, $\mathbb{F}_{2}[x]/M_p(x)$ was used for performing computations in~$\mathbb{F}_{2^{p-1}}$, when $p$ is a prime such that 2 is a primitive element in $\mathbb{F}_p$. In addition, Blaum {\em et al.}~\cite{blaum1996mds,blaum2002evenodd} discussed the rings $\mathbb{F}_{2}[x]/M_p(x)$ in detail.

We will represent each column in a augmented array of $\textsf{EVENODD}(p,k,r;\textbf{g}(k))$ by a polynomial in $R_p$, so that a $p\times n$ array is identified with an $n$-tuple
\begin{equation}
(a'_0(x),a'_1(x),\cdots,a'_{k+r-1}(x))
\label{eq:tuple}
\end{equation}
in $R_{p}^{k+r}$, where $n=k+r$. 
Under this representation, the augmented array of $\textsf{EVENODD}(p,k,r;\textbf{g}(k))$ can be defined in terms of a Vandermonde matrix.

In the $p\times (k+r)$ array, the $p$ bits $a'_{0,j},a'_{1,j},\ldots,a'_{p-1,j}$ in column $j$ can be represented as a polynomial
$$a'_j(x)=a'_{0,j}+a'_{1,j}x+\cdots+a'_{p-1,j}x^{p-1}$$
for $j=0,1,\ldots,k+r-1$. The first $k$ polynomials $a'_0(x),a'_1(x),\ldots,a'_{k-1}(x)$ are called \emph{information polynomials}, and the last $r$ polynomials $a'_{k}(x),a'_{k+1}(x),\ldots,a'_{k+r-1}(x)$ are the \emph{parity polynomials}. The parity bit of augmented array of $\textsf{EVENODD}(p,k,r;\textbf{g}(k))$ defined in \eqref{eq:evenodd1} is equivalent to the following equation over the ring $R_p$
\begin{align}
\begin{bmatrix}
a'_k(x) & \cdots & a'_{k+r-1}(x) 
\end{bmatrix} 
= \begin{bmatrix}
a'_0(x)  & \cdots & a'_{k-1}(x) \\
\end{bmatrix}\cdot
\mathbf{V}_{k\times r}(\textbf{g}(k)),
\label{EVENODD:algebraic}
\end{align}
where $\mathbf{V}_{k\times r}(\textbf{g}(k))$ is the $k\times r$ Vandermonde matrix
\begin{equation}
\mathbf{V}_{k\times r}(\textbf{g}(k)):=
\begin{bmatrix}
 1&x^{g(0)}&\cdots & x^{(r-1)g(0)} \\
 1&x^{g(1)}&\cdots & x^{(r-1)g(1)} \\
  \vdots &\vdots&\ddots & \vdots\\
 1&x^{g(k-1)}& \cdots & x^{(r-1)g(k-1)} \\
 \end{bmatrix}
 \label{eq:vand}
 \end{equation}
and additions and multiplications in the above calculations are performed in $R_p$. \eqref{EVENODD:algebraic} can be verified as follows:
\begin{eqnarray}
a_{k+\ell}'(x)=\sum_{i=0}^{p-1}a_{i,k+\ell}'x^i=\sum_{j=0}^{k-1}a_j(x)x^{\ell  g(j)}=	\sum_{j=0}^{k-1}\sum_{i'=0}^{p-1}a_{i',j}x^{i'+\ell g(j)}=	\sum_{i'=0}^{p-1}\sum_{j=0}^{k-1}a_{i',j}x^{i'+\ell  g(j)}.
\label{eq:sum} 
\end{eqnarray}
Let $i'=i-\ell g(j)$, we have 
$$a_{i,k+\ell}'=\sum_{j=0}^{k-1}a_{i-\ell g(j),j},$$
which is the same as \eqref{eq:evenodd1}.
In other words, each parity column in the augmented array of EVENODD  codes is obtained by adding some cyclically shifted version of the information columns. 

Recall that $a'_{j}(x)$ is a polynomial over $R_p$ for $j=0,1,\ldots,k+r-1$. When we reduce a polynomial $a'_{j}(x)$ modulo $M_p(x)$, it means that we replace the coefficient $a'_{i,j}$ with $a'_{i,j}+a'_{p-1,j}$ for $i=0,1,\ldots,p-2$. When $j=0,1,\ldots,k$, we have that $a'_{p-1,j}=0$. If we reduce $a'_{j}(x)$ modulo $M_p(x)$, we obtain $a'_{j}(x)$ itself, of which the coefficients are the bits of column $j$ of $\textsf{EVENODD}(p,k,r;\mathbf{g}(k))$, for $j=0,1,\ldots,k$. Recall that the coefficients of $a'_{k+\ell}(x)$ for $\ell=1,2,\ldots,r-1$ are computed by \eqref{eq:evenodd1}. If we reduce $a'_{k+\ell}(x)$ modulo $M_p(x)$, i.e., replace the coefficients $a'_{i,k+\ell}$ for $i=0,1,\ldots,p-2$ by $a'_{i,k+\ell}+a'_{p-1,k+\ell}$, which are $a_{i,k+\ell}$ that are bits in column $k+\ell$ of $\textsf{EVENODD}(p,k,r;\mathbf{g}(k))$.
In fact, we have shown how to convert augmented array of $\textsf{EVENODD}(p,k,r;\mathbf{g}(k))$ into original array of $\textsf{EVENODD}(p,k,r;\mathbf{g}(k))$.


By Proposition \ref{prop:shorten}, we can obtain the augmented array of $\textsf{RDP}(p,k,r;\mathbf{g}(k+1))$ by multiplying
\[
(b_0(x),  \cdots,  b_{k-1}(x),\sum_{j=0}^{k-1}b_{j}(x)) 
\begin{bmatrix}
\mathbf{I}_{k+1} & \mathbf{V}_{(k+1)\times r}(\mathbf{g}(k+1))
\end{bmatrix},
\]
and removing the $k+2$-th component, which is always equal to zero, in the resultant product. If we arrange all coefficients in the polynomials with degree strictly less than $p-1$, we get the original $(p-1)\times(k+r)$ array of $\textsf{RDP}(p,k,r;\mathbf{g}(k+1))$.



When $0\leq g(i) \leq k-1$ for $i=0,1,\ldots,k-1$, the MDS property condition of $\textsf{EVENODD}(p,k,r;\textbf{g}(k))$ is the same as that of the extended EVENODD codes \cite{blaum1996mds}, and the MDS property condition of $r\leq 8$ and $r\geq 9$ was given in \cite{blaum1996mds} and \cite{Hou2016On}, respectively. Note that the MDS property condition depends on that 2 is a primitive element in $\mathbb{F}_p$. 
This is the reason of the assumption of primitivity of 2 in $\mathbb{F}_p$.
In the following of the paper, we assume that $0\leq g(i) \leq k-1$ with $i=0,1,\ldots,k-1$ for $\textsf{EVENODD}(p,k,r;\textbf{g}(k))$ and $0\leq g(i) \leq k$ with $i=0,1,\ldots,k$ for $\textsf{RDP}(p,k,r;\textbf{g}(k+1))$, and the proposed $\textsf{EVENODD}(p,k,r;\textbf{g}(k))$ and $\textsf{RDP}(p,k,r;\textbf{g}(k+1))$ are MDS codes. We will focus on the erasure decoding for these two codes.

When some columns of $\textsf{EVENODD}(p,k,r;\textbf{g}(k))$ are erased, we assume that the number of erased information columns is no larger than the number of continuous surviving parity columns. Note that one needs to recover the failure columns by downloading $k$ surviving  columns. First, we represent the downloaded $k$ columns by some information polynomials and continuous parity polynomials. Then, we can subtract all the downloaded information polynomials from the parity polynomials to obtain a Vandermonde linear system. 
Although $\textsf{EVENODD}(p,k,r;\textbf{g}(k))$ can be described by the $k\times r$ Vandermonde matrix given in~\eqref{eq:vand} over $\mathbb{F}_2[x]/M_p(x)$ and we can solve the Vandermonde linear system over $\mathbb{F}_2[x]/M_p(x)$ to recover the failure columns, it is more efficient to solve the Vandermonde linear system over $R_p$.
First, we will show in the next section that we can first perform calculation over $R_p$ and then reduce the results modulo $M_p(x)$ in the decoding process. An efficient decoding algorithm to solve Vandermonde linear system over $R_p$ based on LU factorization of Vandermonde matrix is then proposed in Section \ref{sec:dec}.

\section{Vandermonde Matrix over $R_p$}
\label{sec:vand}
Before we focus on the efficient decoding method of $\textsf{EVENODD}(p,k,r;\textbf{g}(k))$ and $\textsf{RDP}(p,k,r;\textbf{g}(k+1))$, we first present some properties of Vandermonde matrix. As the decoding algorithm hinges on a quick method in solving a Vandermonde system of equations over $R_p$, we  discuss some properties of the linear system of Vandermonde matrix over $R_p$ in this section.

Let $\mathbf{V}_{r\times r}(\mathbf{a})$ be an $r\times r$ Vandermonde matrix
\begin{equation}
\mathbf{V}_{r\times r}(\mathbf{a}):=
\begin{bmatrix}
 1&x^{a_1}&\cdots & x^{(r-1)a_1} \\
 1&x^{a_2}&\cdots & x^{(r-1)a_2} \\
  \vdots &\vdots&\ddots & \vdots\\
 1&x^{a_r}& \cdots & x^{(r-1)a_r} \\
 \end{bmatrix},
\label{eq:square-van}
\end{equation}
where $a_1,\ldots,a_r$ are distinct integers such that the difference of each pair of them is relatively prime to $p$. The entries of $\mathbf{V}_{r\times r}(\mathbf{a})$ are considered as polynomials in $R_p$. We investigate the action of multiplication over $\mathbf{V}_{r\times r}(\mathbf{a})$ by defining the function $F:R_{p}^{r}\rightarrow R_{p}^{r}$:
\[
F(\mathbf{u}):=\mathbf{u}\mathbf{V}_{r\times r}(\mathbf{a})
\]
for $\mathbf{u}=(u_1(x),\ldots,u_r(x))\in R_p^r$. Obviously, $F$ is a homomorphism of abelian group and we have $F(\mathbf{u}+\mathbf{u'})=F(\mathbf{u})+F(\mathbf{u'})$ for $\mathbf{u},\mathbf{u'}\in R_p^r$.

The function $F$ is not surjective. If a vector $\mathbf{v}=(v_1(x),v_2(x),\ldots,v_{r}(x))$ is equal to $F(\mathbf{u})$ for some $\mathbf{u}\in R_p$, it is necessary that
\begin{equation}
v_1(1)=v_2(1)=\cdots=v_{r}(1).
\label{eq:terms}
\end{equation}
This is due to the fact that  each polynomial $v_j(x)$ is obtained by adding certain cyclically shifted version of $u_i(x)$'s. In other words, if $\mathbf{v}$ is in the image of $F$, then either there are even number of nonzero terms in  all $v_i(x)$, or there are odd number of nonzero terms in all $v_i(x)$ for $1\le i\le r$.

The function $F$ is also not injective. We can see this by observing that if we add the polynomial $M_p(x)$ to a component of $\mathbf{u}$, for example,  adding $M_p(x)$ to $u_i(x)$, then 
\begin{align*}
F(\mathbf{u}+(\underbrace{0,\ldots,0}_{i-1},M_p(x),0,\ldots,0))
=F(\mathbf{u})+(M_p(x),\ldots,M_p(x)).
\end{align*}
Hence, if we add $M_p(x)$ to two distinct components of input vector $\mathbf{u}$, then the value of $F(\mathbf{u})$ does not change. 
We need the following lemma before discussing the properties of the Vandermonde linear system over $R_p$.

\begin{lemma}\cite[Lemma 2.1]{blaum1996mds}
Suppose that $p$ is an odd number and $d$ is relatively prime to $p$, then $1+x^d$ and $M_p(x)$ are coprime in $\mathbb{F}_2[x]$, and $x^i$ and $M_p(x)$ are relatively prime in $\mathbb{F}_2[x]$ for any positive integer $i$.
\label{lm:coprime}
\end{lemma}
If the vector $\mathbf{v}$ satisfies  \eqref{eq:terms},  in the next theorem, we show that there are many vectors $\mathbf{u}$ such that $F(\mathbf{u})=\mathbf{v}$.

\begin{theorem}
Let $a_1,a_2,\ldots,a_r$ be $r$ integers such that the difference $a_{i_1}-a_{i_2}$ is relatively prime to $p$ for all pair of distinct indices $1\leq i_1<i_2\leq r$. The image of $F$ consists of all vectors $\mathbf{v}\in R_p^r$ that satisfy the condition \eqref{eq:terms}. For all vectors $\mathbf{u}$ satisfying 
\begin{equation}
\mathbf{u}\mathbf{V}_{r\times r}(\mathbf{a})=\mathbf{v}\mod 1+x^p,
\label{eq:vanrp}
\end{equation}
they are congruent to each other modulo $M_p(x)$.
\label{thm:unique}
\end{theorem}
\begin{proof}
Suppose that $v_1(x),\ldots,v_{r}(x)$ are polynomials in $R_p$ satisfying \eqref{eq:terms}. We want to show that the vector $\mathbf{v}=(v_{1}(x),\ldots,v_{r}(x))$ is in the image of $F$. We first consider the case that $v_1(1)=v_2(1)=\cdots=v_{r}(1)=0$.
 Since $1+x$ and $M_p(x)$ are relatively prime polynomials, by Chinese remainder theorem, we have an isomorphism
 \[
 \theta(f(x))=(f(x)\bmod 1+x,f(x)\bmod M_p(x))
 \]
 defined for $f(x)\in R_p$. The inverse of $\theta$ is given by
 \[
 \theta^{-1}(a(x),b(x))
=M_p(x)a(x)+(1+M_p(x))b(x) \bmod 1+x^p,
 \]
where $a(x)\in\mathbb{F}_2[x]/(1+x)$ and $b(x)\in\mathbb{F}_2[x]/M_p(x)$. We thus have a decomposition of the ring $R_p$ as a direct sum of $\mathbb{F}_2[x]/(1+x)$ and $\mathbb{F}_2[x]/M_p(x)$. It suffices to investigate the action of multiplication over $\mathbf{V}_{r\times r}(\mathbf{a})$ by considering 
 \begin{equation}
\mathbf{u}\mathbf{V}_{r\times r}(\mathbf{a})=\mathbf{v}\mod 1+x, \text{ and }
\label{eq:factor1}
 \end{equation}
 \begin{equation}
\mathbf{u}\mathbf{V}_{r\times r}(\mathbf{a})=\mathbf{v}\mod M_p(x).
\label{eq:factor2}
\end{equation}
Note that \eqref{eq:factor1} is equivalent to 
 \[
(\mathbf{u} \bmod (1+x)) \cdot (\mathbf{V}_{r\times r}(\mathbf{a})\bmod (1+x))=\mathbf{v}\bmod (1+x).
 \] 
Also $\mathbf{V}_{r\times r}(\mathbf{a})\bmod (1+x)$ is an $r\times r$ all one matrix and $\mathbf{v}\bmod (1+x)=\mathbf{0}$ because $v_i(1)=0,\ 1\le i\le r$. It is sufficient to find $r$ components of solution $\mathbf{u}'$ from binary field such that their summation is zero. Therefore, there are many solutions $\mathbf{u}'$ and each solution satisfies that the number of one among all the components of  $\mathbf{u}'$ is an even number.

 For \eqref{eq:factor2}, we need to show that the determinant of $\mathbf{V}_{r\times r}(\mathbf{a})$ is invertible modulo $M_p(x)$. Since the determinant of $\mathbf{V}_{r\times r}(\mathbf{a})$ is\footnote{Since $-1$ is the same as $1$ in $\mathbb{F}_2$, we replace $-1$ with $1$ in this work and addition is the same as subtraction.}
 \[
 \det(\mathbf{V}_{r\times r}(\mathbf{a}))=\prod_{i_1<i_2}(x^{a_{i_1}}+x^{a_{i_2}}),
 \]
 we need to show that $x^{a_{i_1}}+x^{a_{i_2}}$ and $M_p(x)$ are relatively prime polynomials in $\mathbb{F}_2[x]$, for all pairs of distinct $(i_1,i_2)$. We first factorize $x^{a_{i_1}}+x^{a_{i_2}}$ in the form $x^{a_{i_1}}+x^{a_{i_2}}=x^{a_{i_1}}(1+x^{d})$, 
 and by assumption, $d$ is an integer which is coprime with $p$. This problem now reduces to showing that (i) $1+x^d$ and $M_p(x)$ are relatively prime in $\mathbb{F}_2[x]$ whenever $\gcd(d,p)=1$, and (ii) $x^{\ell}$ and $M_p(x)$ are relatively prime in $\mathbb{F}_2[x]$ for all integer $\ell$. We can show this by Lemma \ref{lm:coprime}.
We can thus solve the equation in \eqref{eq:factor2}, say, by Cramer's rule, to obtain the unique solution $\mathbf{u}''$. 
 After obtaining the solutions $u'_i(x)\in \mathbb{F}_2[x]/(1+x)$ and $u''_i(x)\in \mathbb{F}_2[x]/M_p(x)$ in \eqref{eq:factor1} and \eqref{eq:factor2} respectively for all $i$, we can calculate the solution via the isomorphism $\theta^{-1}$ 
\begin{eqnarray}
\theta^{-1}(u'_i(x),u''_i(x))
=M_p(x)u'_i(x)+(1+M_p(x))u''_i(x) \bmod 1+x^p. \nonumber
\end{eqnarray}
Therefore, the solutions of $\mathbf{u}\in R_p^r$ in \eqref{eq:vanrp} are 
\begin{equation}
((1+M_p(x))u''_1(x)+\epsilon_1 M_p(x),\cdots,(1+M_p(x))u''_r(x)+\epsilon_r M_p(x)),
\label{eq:solution}
\end{equation}
where $\epsilon_i$ is equal to 0 or 1 for all $i$ and the number of ones is an even number. That is to say, there are many solutions in \eqref{eq:vanrp} and all the solutions after reducing modulo $M_p(x)$ is unique and is the solution in \eqref{eq:factor2}.

When $v_1(1)=v_2(1)=\cdots=v_r(1)=1$, similar argument can be applied to find the solution. The only difference between this solution and \eqref{eq:solution} is that the number of ones among all $\epsilon_i$ in this solution is odd. This completes the proof.
\end{proof}

From the above theorem, whenever the vector $\mathbf{v}$ satisfies the condition \eqref{eq:terms}, we can decode one solution of $\mathbf{u}$ in \eqref{eq:vanrp}. 
Recall that, for augmented array of EVENODD codes, every component  in $\mathbf{u}$ has been added zero coefficient of the term with degree $p-1$. Hence, each component of the real solution is with at most degree $p-2$. By the theorem,  reducing $\mathbf{u}$ modulo  $M_p(x)$ gives us the final solution. 
Therefore, to solve $\mathbf{u}$ in \eqref{eq:factor2}, we can first solve $\mathbf{u}$ over $R_p$ and then take the modulo $M_p(x)$ for every component of $\mathbf{u}$. This will be demonstrated in next section.

\section{Efficient Decoding of Vandermonde System over $R_p$}
\label{sec:dec} 
In this section, we will present an efficient decoding method of the Vandermonde system over $R_p$ based on LU factorization of the Vandermonde matrix.
 
\subsection{Efficient Division by $1+x^d$}

We want to first present two decoding algorithms for performing division by $1+x^d$ which will be used in the decoding process, where $d$ is a positive integer that is coprime with $p$ and all operations are performed in the ring $R_p=\mathbb{F}_2[x]/(1+x^p)$. Given the equation 
\begin{equation}
(1+x^d)g(x)=f(x) \bmod 1+x^p,
\label{eq;div1}
\end{equation} 
where $d$ is a positive integer such that $\gcd(d,p)=1$ and $f(x)$ has even number of nonzero terms. One method to compute $g(x)$ is given in Lemma 8 in \cite{Hou2017new}, which is summarized as follows.

\begin{lemma}\cite[Lemma 8]{Hou2017new}
The coefficients of $g(x)$ in \eqref{eq;div1} can be computed by  
\begin{eqnarray}
&g_{p-1}=0, g_{p-d-1}=f_{p-1}, g_{d-1}=f_{d-1}, \nonumber\\ 
&g_{(p-(i+1)d-1) }=f_{(p-id-1) }+g_{(p-id-1) }  \text{ for } i=1,\ldots,p-3, \nonumber
\end{eqnarray}
where $g(x)=\sum_{i=0}^{p-1}g_i x^i,f(x)=\sum_{i=0}^{p-1}f_i x^i$.
\label{lm:div}
\end{lemma}
Although computing the division in \eqref{eq;div1} by Lemma \ref{lm:div} only takes $p-3$ XORs, we do not know whether the resulting polynomial $g(x)$ has even number of nonzero terms or not. In solving the Vandermonde linear system in the next subsection, we need to compute many divisions of the form in~\eqref{eq;div1}. If we do not require that the solved polynomial $g(x)$ should have even number of nonzero terms, then we can employ Lemma~\ref{lm:div} to solve the division. Otherwise, Lemma~\ref{lm:div} is not applicable to such division. Therefore, we need the following lemma that can compute the division when $g(x)$ is required to have even number of nonzero terms.
\begin{lemma}\cite[Lemma 13]{Hou2016BASIC}
Given the equation in \eqref{eq;div1}, we can compute the coefficient $g_0$ by
\begin{equation}
g_0=f_{2d}+f_{4d}+\cdots+f_{(p-1)d},
\label{eq:invcmpt}
\end{equation}
and the other coefficients of $g(x)$ can be iteratively computed by
\begin{equation}
g_{d\ell }=f_{d\ell}+g_{d(\ell-1)} \text{ for } \ell=1,2,\ldots,p-1.
\label{eq:invcmpt2}
\end{equation}
\label{lm:div1}
\end{lemma}
\vspace{-0.3cm}
Note that the subscripts in Lemma \ref{lm:div1} are taken modulo $p$. As $\gcd(d,p)=1$, we have that
\[
\{0,d,2d,\cdots,(p-1)d\}=\{0,1,2,\cdots,p-1\} \bmod p.
\]
Therefore, we can compute all the coefficients of $g(x)$ by Lemma~\ref{lm:div1}. The result in Lemma \ref{lm:div1} has been observed in \cite{Hou2014New,Hou2016BASIC}. We can check that the computed polynomial $g(x)$ satisfies $g(1)=0$, by adding the equation in \eqref{eq:invcmpt} and the equations in \eqref{eq:invcmpt2} for $\ell=2,4,\ldots,p-1$. The number of XORs required in computing the division by Lemma \ref{lm:div1} is $(3p-5)/2$.


For the same parameters, the computed $g(x)$ by Lemma \ref{lm:div} is equal to either the computed $g(x)$ by Lemma \ref{lm:div1} or the summation of $M_p(x)$ and the computed $g(x)$ by Lemma \ref{lm:div1}, depending on whether the computed $g(x)$ by Lemma \ref{lm:div} has even number of nonzero terms or not. When solving a division in \eqref{eq;div1}, we prefer the method in Lemma \ref{lm:div} if there is no requirement that the resulting polynomial $g(x)$ should have even number of nonzero terms, as the method in Lemma~\ref{lm:div} involves less XORs.

\subsection{LU Method of Vandermonde Systems over $R_p$}
\label{sec:lu}
The next theorem is the core of the fast LU method for solving Vandermonde system of equations $\mathbf{v}=\mathbf{u}\mathbf{V}_{r\times r}(\mathbf{a})$ which is based on the LU decomposition of a Vandermonde matrix given in~\cite{yang2005lu}. 
\begin{theorem}~\cite{yang2005lu}
For positive integer $r$, the square Vandermonde matrix $\mathbf{V}_{r\times r}(\mathbf{a})$ defined in \eqref{eq:square-van} can be factorized into
\[
\mathbf{V}_{r\times r}(\mathbf{a})=\mathbf{L}_{r}^{(1)}\mathbf{L}_{r}^{(2)}\cdots \mathbf{L}_{r}^{(r-1)}\mathbf{U}_{r}^{(r-1)}\mathbf{U}_{r}^{(r-2)}\cdots\mathbf{U}_{r}^{(1)},
\]
where $\mathbf{U}_r^{(\ell)}$ is the upper triangular matrix
\begin{align*}
\mathbf{U}_r^{(\ell)}:=
\left[\begin{array}{c|cccccc}
                             \mathbf{I}_{r-\ell-1} & & &  \text{\LARGE 0} & & &  \\
                             \hline
                               & 1 & x^{a_1} & 0 & \cdots & 0 & 0\\
                               & 0 & 1 & x^{a_2} & \cdots & 0 & 0\\
                                \text{\LARGE 0}&  \vdots & \vdots & \vdots & \ddots & \vdots & \vdots \\
                              & 0 & 0 & 0 & \cdots & 1 & x^{a_\ell} \\  
                              & 0 & 0 & 0 & \cdots & 0 & 1   
                           \end{array}\right]
\end{align*}
and $\mathbf{L}_r^{(\ell)}$ is the lower triangular matrix
\[
\small
\left[\begin{array}{c|ccccc}
                             \mathbf{I}_{r-\ell-1} & &  &\text{\large 0}  & & \\
                             \hline
                               & 1  & 0  & \cdots & 0 & 0  \\
                               & 1  & x^{a_{r-\ell+1}}+x^{a_{r-\ell}}  & \cdots & 0 & 0  \\
                                \text{\large 0}&  \vdots & \vdots  & \ddots & \vdots & \vdots \\
                               & 0 & 0  & \cdots & x^{a_{r-1}}+x^{a_{r-\ell}} & 0 \\  
                               & 0 & 0 & \cdots & 1 & x^{a_{r}}+x^{a_{r-\ell}}   
                           \end{array}\right]
\]
for $\ell=1,2,\ldots,r-1$.
\label{thm:factor}
\end{theorem}
For example, the Vandermonde matrix $\mathbf{V}_{3\times 3}(1,x,x^4)$ can be factorized as
\begin{align*}
\mathbf{L}_{3}^{(1)}\mathbf{L}_{3}^{(2)}\mathbf{U}_{3}^{(2)}\mathbf{U}_{3}^{(1)}=\begin{bmatrix}1 & 0 & 0\\
0 & 1 & 0 \\
0 & 1 & x^{4}+x 
\end{bmatrix}\begin{bmatrix}1 & 0 & 0\\
1 & x+1 & 0 \\
0 & 1 & x^4+1 
\end{bmatrix}
\begin{bmatrix}1 & 1 & 0\\
0 & 1 & x \\
0 & 0 & 1 
\end{bmatrix}\begin{bmatrix}1 & 0 & 0\\
0 & 1 & 1 \\
0 & 0 & 1 
\end{bmatrix}.
\end{align*}
Based on the factorization in Theorem \ref{thm:factor}, we have a fast algorithm for solving a Vandermonde system of linear equations. Given the matrix $\mathbf{V}_{r\times r}(\mathbf{a})$ and a row vector $\mathbf{v}=(v_1(x),\ldots,v_r(x))$, we can solve the linear system $\mathbf{u}\mathbf{V}_{r\times r}(\mathbf{a})=\mathbf{v}$ in \eqref{eq:vanrp} by solving
\begin{equation}
\mathbf{u}\mathbf{L}_{r}^{(1)}\mathbf{L}_{r}^{(2)}\cdots \mathbf{L}_{r}^{(r-1)}\mathbf{U}_{r}^{(r-1)}\mathbf{U}_{r}^{(r-2)}\cdots \mathbf{U}_{r}^{(1)}=\mathbf{v}.
\label{eq:factor}
\end{equation}
As the inversion of each of the upper and lower triangular matrices can be done efficiently, we can solve for $\mathbf{u}$ by inverting $2(r-1)$ triangular matrices.

\begin{algorithm}
\caption{Solving a Vandermonde linear system.} \label{alg:Vand}
{\bf Inputs:}
\quad positive integer $r$, odd integer $p$, integers $a_1,a_2,\ldots, a_{r}$, and $\mathbf{v} = (v_1(x),v_2(x),\ldots, v_r(x))\in R_p^r$.\\
{\bf Output:}
\quad $\mathbf{u} = (u_1(x),\ldots, u_r(x))$ that satisfies $\mathbf{u}\mathbf{V}_{r\times r}(\mathbf{a})=\mathbf{v}$. 
\begin{algorithmic}[1]
\REQUIRE  $v_1(1)=v_2(1)=\cdots=v_r(1)$, and $\gcd(a_{i_1}-a_{i_2},p)=1$ for all $1\leq i_1<i_2\leq r$.
\STATE $\mathbf{u} \gets \mathbf{v}$.
\FOR{ $i$ from $1$ to $r-1$}
\FOR{ $j$ from $r-i+1$ to $r$}
\STATE $u_{j}(x) \gets  u_{j}(x) +u_{j-1}(x) x^{a_{i+j-r}}$
\ENDFOR
\ENDFOR
\FOR{ $i$ from $r-1$ down to $1$}
\STATE Solve $g(x)$ from $(x^{a_{r}} +x^{a_{r-i}})g(x)= u_{r}(x)$ by Lemma~\ref{lm:div1} or Lemma~\ref{lm:div} (only when $i=1$) \\
 $u_{r}(x) \gets g(x)$ 
\FOR{ $j$ from $r-1$ down to $r-i+1$}
\STATE  Solve $g(x)$ from $( x^{a_{j}} +x^{a_{r-i}} )g(x)=  (u_{j}(x) +u_{j+1}(x))$ by  Lemma~\ref{lm:div1} or Lemma~\ref{lm:div} (only when $i+j=r+1$)\\ $u_{j}(x) \gets g(x)$
\ENDFOR
\STATE  $u_{r-i}(x) \gets u_{r-i}(x) +u_{r-i+1}(x)$
\ENDFOR
\RETURN{$\mathbf{u}=(u_1(x),\ldots,u_r(x))$}
\end{algorithmic}
\end{algorithm}

The procedure of solving a Vandermonde system of linear equations is given in Algorithm \ref{alg:Vand}. In Algorithm \ref{alg:Vand},  steps 2 to 4 are forward additions that require $r(r-1)/2$ additions and $r(r-1)/2$ multiplications. Steps 5 to 9 are backward additions, and require $r(r-1)/2$ additions and $r(r-1)/2$ divisions by factors of the form $x^{a_j}+x^{a_{r-i}}$. Division by $x^{a_j}+x^{a_{r-i}}$ is done by invoking the method in Lemma \ref{lm:div} or Lemma \ref{lm:div1}. We may compute all the divisions by Lemma~\ref{lm:div1}. However, some of the division can be computed by  Lemma \ref{lm:div}, where the computational complexity can be reduced. 

\begin{theorem}
Algorithm \ref{alg:Vand} outputs $\mathbf{u}$ that is one of the solution to $\mathbf{u} \mathbf{V}_{r\times r}(\mathbf{a})=\mathbf{v}$ over $R_p$. Furthermore, $g(x)$ in step 6 when $i=1$, and in step 8 when $i+j=r+1$ and $i\in \{2,3,\ldots,r-1\}$ can be computed by Lemma~\ref{lm:div} to reduce computation complexity. 
\label{thm:vand}
\end{theorem}
\begin{proof}
First, we want to show that Algorithm \ref{alg:Vand} implements precisely the matrix multiplications in \eqref{eq:factor}. Consider the linear equations $\mathbf{u}\mathbf{U}_{r}^{(i)}=\mathbf{v}$ for $i=1,2,\ldots,r-1$. According to the upper triangular matrix $\mathbf{U}_{r}^{(i)}$ in Theorem \ref{thm:factor}, we can obtain the relation between $\mathbf{u}$ and $\mathbf{v}$ as 
\begin{align*}
&u_j(x)=v_{j}(x) \text{ for } j=1,\ldots,r-i,\\
&x^{a_{i+j-r}}u_{j-1}(x)+u_{j}(x)=v_{j}(x) \text{ for } j=r-i+1,\ldots,r.
\end{align*}
We can observe from Algorithm \ref{alg:Vand} that steps 1 to 4 solve $\mathbf{u}$ from $\mathbf{u}\mathbf{U}_{r}^{(r-1)}\mathbf{U}_{r}^{(r-2)}\cdots \mathbf{U}_{r}^{(1)}=\mathbf{v}$, and we denote the solved $r$ polynomials as $\mathbf{v}'=(v'_1(x),\ldots,v'_r(x))$. Consider the equations $\mathbf{u}\mathbf{L}_{r}^{(i)}=\mathbf{v}'$. According to the lower triangular matrix $\mathbf{L}_{r}^{(i)}$ in Theorem \ref{thm:factor}, the relation between $\mathbf{u}$ and $\mathbf{v}'$ is as follows:
\begin{eqnarray}
&& u_j(x)=v'_{j}(x) \text{ for } j=1,\ldots,r-i-1,\nonumber \\ 
&& u_{r-i}(x)+u_{r-i+1}(x)=v'_{r-i}(x), \label{eq:eq1}\\
&& (x^{a_{j}}+x^{a_{r-i}})u_j(x)+u_{j+1}(x)=v'_j(x)  \text{ for } j=r-i+1,\ldots,r-1, \label{eq:eq2}\\
&& (x^{a_{r}}+x^{a_{r-i}})u_r(x)=v'_{r}(x).\label{eq:eq3}
\end{eqnarray}
It is easy to see that step 6, step 8 and step 9 solves $u_r(x)$ from \eqref{eq:eq3}, $u_j(x)$ from \eqref{eq:eq2} and $u_{r-i}(x)$ from \eqref{eq:eq1} respectively. We thus obtain that steps 5 to 9 solve $\mathbf{u}$ from $\mathbf{u}\mathbf{L}_{r}^{(1)}\mathbf{L}_{r}^{(2)}\cdots \mathbf{L}_{r}^{(r-1)}=\mathbf{v}'$. Therefore, Algorithm \ref{alg:Vand} precisely computes $\mathbf{u}$ from the matrix multiplication in  \eqref{eq:factor}.

Note that we need to compute $r(r-1)/2$ divisions (solving $g(x)$) in steps 6 and 8 in Algorithm \ref{alg:Vand}, and all divisions can be solved by Lemma \ref{lm:div} or Lemma \ref{lm:div1}. To solve all divisions by Lemma \ref{lm:div} or Lemma \ref{lm:div1}, all polynomials $u_r(x)$ and $u_j(x)+u_{j+1}(x)$ in steps 6 and 8 must have even number of nonzero terms, which is a requirement in computing the divisions. That is $u_r(1)=0$ and $u_j(1)+u_{j+1}(1)=0$ when we calculate $g(x)$ in these steps. Next we show this requirement is ensured during the process of the algorithm.

Consider two cases: $v_i(1)=0$ for all $i$ and $v_i(1)=1$ for all $i$. If $v_i(1)=0$ for all $i$, then we have $u_i(1)=0$ for all $i$ after the double for loops between steps 2 to 4. Hence $u_r(1)=0$ in step 6 when $i=r-1$. In step 6, we need to ensure that all $u_r(x)$ produced  satisfy $u_r(1)=0$.  Since all $g(x)$, hence $u_r(x)$, generated by Lemma~ \ref{lm:div1} have such property, we only need to consider the case that $g(x)$ is generated by Lemma~ \ref{lm:div}, i.e., when $i=1$.   When $i=1$, the  $u_r(x)$ assigned by $g(x)$ will not be invoked in computing the division in steps 7 and 8 as  $r-i+1=r$ which is larger than the initial value $r-1$ of the loop  in step 7. Hence, there is no need to run the for loop in step 7 to step 8.  

Similarly, we need to ensure that $u_j(1)-u_{j+1}(1)=0$ for $j=r-1, r-2,\ldots, r-i-1$. It is sufficient to ensure that $u_j(1)=0$ and $u_{j+1}(1)=0$. Since all $g(x)$, hence $u_j(x)$, generated by Lemma~ \ref{lm:div1} already have had such property, in step 8, we only need to consider the case that $g(x)$ is generated by Lemma~ \ref{lm:div}, i.e., when $i+j=r+1$. Next, we prove that when $i+j=r+1$, the $u_j(x)$ assigned by $g(x)$ in step 8 are not used again to solve the division in step 8. Let $i=t$ ($i=r-1,\ldots,2$). Note that, when $i+j=r+1$, the algorithm is in the final iteration of the for loop in step 7. Hence, we need to prove that $u_j(x)=u_{r-t+1}(x)$ will not be used in the iterations $i<t$.  When $i<t$ in step 7, the last iteration $j=r-i+1>r-t+1$ such that $u_{r-t+1}(x)$ will not be used in the calculation involving $u_j(x)$ and $u_{j+1}(x)$ in step 8.

If $v_i(1)=1$ for all $i$, then after the double for loops between steps 2 to 4, we have
\[
u_1(1)=1, \text{ and } u_2(1)=u_3(1)=\cdots =u_r(1)=0.
\]
In this case, we only need to show that $u_1(x)$ has never been used in the calculation of $g(x)$ in step 8. Note that, in the last iteration of step 7, $j=r-i+1$ has never gone down to $1$ since $i\le r-1$. Hence, $u_1(x)$ has never been used in step 8.
\end{proof}




The next theorem shows the computational complexity in Algorithm \ref{alg:Vand}.
\begin{theorem}
The computation complexity in Algorithm \ref{alg:Vand} is at most
\begin{equation}
r(r-1)p+(r-1)(p-3)+(r-1)(r-2)(3p-5)/4.
\label{eq:lucom}
\end{equation}
\label{thm:lucom}
\end{theorem}
\begin{proof}
In Algorithm \ref{alg:Vand}, there are $r(r-1)$ additions that require $r(r-1)p$ XORs, $r(r-1)/2$ multiplications that require no XORs (only cyclic shift applied) and $r(r-1)/2$ divisions that require $(r-1)(p-3)+(r-1)(r-2)(3p-5)/4$ XORs ($r-1$ divisions are computed by Lemma \ref{lm:div} and the other divisions are computed by Lemma \ref{lm:div1}).
Therefore, the total computation of Algorithm \ref{alg:Vand} is at most \eqref{eq:lucom}.
\end{proof}
According to Theorem \ref{thm:unique}, we can solve the Vandermonde system over $\mathbb{F}_2[x]/M_p(x)$ by first solving the Vandermonde system over $R_p$ and then reducing the $r$ resulting polynomials by modulo $M_p(x)$. By Theorem~\ref{thm:vand}, we can solve the Vandermonde system over $R_p$ by using the factorization method in Theorem~\ref{thm:factor}. 

\textbf{Remark.} Since in Algorithm \ref{alg:Vand}, when $i=1$, the output $u_r(x)$ is computed from the division in step 6  which is solved by Lemma~\ref{lm:div}, we have that the last coefficient of $u_r(x)$, $g_{p-1}$, is zero. For $i=1,2,\ldots,r-1$, the output $u_{r-i}(x)$ is a summation of $u_{r-i}(x)$ and $u_{r-i+1}(x)$, where $u_{r-i+1}(x)$ is computed in the last iteration  in step 7 for $i=r-1,r-2,\ldots,2$,  and $u_{r-i}(x)$ is computed in the last iteration  in step 7 for $i+1=r-1,r-2,\ldots,2$. Note that the last iteration for each $i$  is solved by Lemma \ref{lm:div}. Thus  the last coefficient of $u_i(x)$ is zero for $i=2,3,\ldots,r-1$. The output $u_1(x)$ is the summation of $u_1(x)$ and $u_2(x)$ by step 9, where $u_1(x)=v_1(x)$ and $u_2(x)$ is computed from the division in step 8 when $i=r-1$, which is solved by Lemma~\ref{lm:div}. Therefore, if the last coefficient of $v_1(x)$ is zero, then the last coefficient of the output $u_1(x)$ is zero. Otherwise, the last coefficient of the output $u_1(x)$ is one. We thus have that the last coefficient of each of the last $r-1$ resulting polynomials is zero. Therefore, it is not necessary to reduce the last $r-1$ resulting polynomials by modulo $M_p(x)$ and we only need to reduce the first resulting polynomial by modulo $M_p(x)$. In the example of $\textsf{EVENODD}(5,3,3;(0,1,4))$, the three components of the returned $\mathbf{u}$ are exactly equal to the three information polynomials of $\textsf{EVENODD}(5,3,3;(0,1,4))$, as the last coefficient of $v_1(x)$ is zero.

Although the LU decoding method of the $r\times r$ Vandermonde linear systems over $R_p$ is also discussed in Theorem 14 of \cite{Hou2016BASIC}, the complexity of the algorithm provided is $\frac{7}{4}r(r-1)p$ which is larger than~\eqref{eq:lucom}. The reason of computation reduction given in ~\eqref{eq:lucom} is as follows. There are $r-1$ divisions in Algorithm \ref{alg:Vand} that are solved with $p-3$ XORs involved for each division, while all $r(r-1)$ divisions are solved with $(3p-5)/2$ XORs involved for each division in \cite{Hou2016BASIC}.


\section{Erasure Decoding of $\textsf{EVENODD}(p,k,r)$ and $\textsf{RDP}(p,k,r)$}

The efficient decoding method of the Vandermonde linear systems over $R_p$ proposed in Section \ref{sec:dec} is applicable to the information column failure and some particular cases with both information failure and parity failure of $\textsf{EVENODD}(p,k,r;\textbf{g}(k))$, $\textsf{RDP}(p,k,r;\textbf{g}(k+1))$ and  Vandermonde array codes such as Blaum-Roth code~\cite{BlaumRoth93}. We first consider the decoding method for $\textsf{EVENODD}(p,k,r;\textbf{g}(k))$.

\subsection{Erasures Decoding of $\textsf{EVENODD}(p,k,r)$}
Suppose that $\gamma$ information columns $e_{1},\ldots,e_{\gamma}$ and $\delta$ parity columns $f_{1},\ldots,f_{\delta}$ are erased with $0\leq e_{1}<\ldots<e_{\gamma}\leq k-1$ and $k+1\leq f_{1}<\ldots<f_{\delta}\leq k+r-1$, where $k\geq \gamma \geq 0$, $r-1\geq \delta \geq 0$ and $\gamma+\delta=\rho\leq r$. Let $f_0=k-1$ and $f_{\delta+1}=k+r-1$, we assume that there exist $\lambda \in \{0,1,\ldots,\delta \}$ such that $f_{\lambda+1}-f_{\lambda}\geq \gamma+1$. We have that the columns $f_{\lambda}+1,\ldots,f_{\lambda}+\gamma$ are not erased. Let $$\mathcal{A}:=\{0,1,\ldots,k-1\}\setminus \{e_{1},e_{2},\ldots,e_{\gamma}\}$$
be a set of indices of the available information columns.
We want to first recover the lost information columns by reading $k-\gamma$ information columns with indices $i_1,i_2,\ldots, i_{k-\gamma}\in\mathcal{A}$, and $\gamma$ parity columns with indices $f_{\lambda}+1, f_{\lambda}+2,\ldots,f_{\lambda}+\gamma$, and then recover the failure parity column by re-encoding the failure parity bits according to  \eqref{EVENODD:diag} for $\ell=f_{1}-k,\ldots,f_{\delta}-k$.

First, we compute the bits of the $\gamma$ parity columns $f_{\lambda}+1, f_{\lambda}+2,\ldots,f_{\lambda}+\gamma$ of the augmented array according to \eqref{eq:augmented1} and \eqref{eq:augmented2} in Lemma \ref{lm:evenodd}. This can be done since column $k$ is not failed. Then, we represent the $k-\gamma$ information columns and $\gamma$ parity columns by $k-\gamma$ information polynomials $a'_i(x)$ as
\begin{equation}
a'_i(x):=a_{0,i}+a_{1,i}x+\cdots+a_{p-2,i}x^{p-2}
\label{eq:infor-poly}
\end{equation}
for $i\in \mathcal{A}$ and $\gamma$ parity polynomials $a'_{f_{\lambda}+j}(x)$
\begin{equation}
a'_{f_{\lambda}+j}(x):=a'_{0,f_{\lambda}+j}+a'_{1,f_{\lambda}+j}x+\cdots+a'_{p-1,f_{\lambda}+j}x^{p-1}
\label{eq:parity-poly}
\end{equation}
for $j=1,2,\ldots,\gamma$. Then, we subtract $k-\gamma$ information polynomials $a'_i(x)$ in \eqref{eq:infor-poly}, $i\in\mathcal{A}$ from the $\gamma$ parity polynomials $a'_{f_{\lambda}+1}(x),\ldots,a'_{f_{\lambda}+\gamma}(x)$ in \eqref{eq:parity-poly}, to obtain $\gamma$ \emph{syndrome polynomial} $\bar{a}_h(x)$ over $R_p$ as
\begin{equation}
\bar{a}_{h}(x)=a'_{f_{\lambda}+h}(x)+\sum_{i\in\mathcal{A}}a'_{i}(x)x^{g(i)\cdot (f_{\lambda}+h-k)},
\label{eq:syndrome}
\end{equation}
for $h=1,2,\ldots,\gamma$. Therefore, we can establish the relation between the syndrome polynomials and the erased information polynomials as follows
\begin{align*}
 \begin{bmatrix}
\bar{a}_1(x) &  \cdots & \bar{a}_{\gamma}(x)
\end{bmatrix}=
\begin{bmatrix}
a'_{e_1}(x)  & \cdots & a'_{e_{\gamma}}(x)
\end{bmatrix}\cdot 
\begin{bmatrix}
x^{g(e_1)(f_{\lambda}+1-k)} & x^{g(e_1)(f_{\lambda}+2-k)} & \cdots & x^{g(e_1)(f_{\lambda}+\gamma-k)} \\
x^{g(e_2)(f_{\lambda}+1-k)} & x^{g(e_2)(f_{\lambda}+2-k)} & \cdots & x^{g(e_2)(f_{\lambda}+\gamma-k)} \\
\vdots & \vdots & \ddots & \vdots \\
x^{g(e_\gamma)(f_{\lambda}+1-k)} & x^{g(e_\gamma)(f_{\lambda}+2-k)} & \cdots & x^{g(e_\gamma)(f_{\lambda}+\gamma-k)} \\
\end{bmatrix}.
\end{align*}
The right-hand side of the above equations can be reformulated as
\begin{align*}
\begin{bmatrix}
x^{g(e_1)(f_{\lambda}+1-k)}a'_{e_1}(x)  & \cdots & x^{g(e_\gamma)(f_{\lambda}+1-k)}a'_{e_{\gamma}}(x)
\end{bmatrix} \mathbf{V}_{\gamma\times \gamma}(\mathbf{e}),
\end{align*}
where $\mathbf{V}_{\gamma\times \gamma}(\mathbf{e})$ is a Vandermonde matrix
\[
\mathbf{V}_{\gamma\times \gamma}(\mathbf{e}):=\begin{bmatrix}
1 & x^{g(e_1)} & \cdots & x^{g(e_1)(\gamma-1)} \\
1 & x^{g(e_2)} & \cdots & x^{g(e_2)(\gamma-1)} \\
\vdots & \vdots & \ddots & \vdots \\
1 & x^{g(e_\gamma)} & \cdots & x^{g(e_\gamma)(\gamma-1)} \\
\end{bmatrix}.
\]
By \eqref{eq:sum}, we have that $a'_{f_{\lambda}+h}(1)=\sum_{j=0}^{k-1}a'_j(1)$ and we thus have 
$$\bar{a}_h(1)=a'_{f_{\lambda}+h}(1)+\sum_{i\in\mathcal{A}}a'_{i}(1)=\sum_{j=0}^{k-1}a'_j(1)+\sum_{i\in\mathcal{A}}a'_{i}(1),$$
which is independent on $h$. Thus, we obtain that $\bar{a}_1(1)=\cdots=\bar{a}_{r}(1)$. We can then obtain the erased information polynomials by first solving the Vandermonde linear systems over $R_p$ by Algorithm \ref{alg:Vand}, cyclic-left-shifting the solved polynomials $x^{g(e_i)(f_{\lambda}+1-k)}a'_{e_i}(x)$ by $g(e_i)(f_{\lambda}+1-k)$ positions for $i=1,2,\ldots,\gamma$, and then reduce $a'_{e_1}(x)$ modulo $M_p(x)$ when $\lambda>0$ according to the remark at the end of Section \ref{sec:lu}. If $\lambda=0$, we have $f_{\lambda}+1-k=0$ and the last coefficient of $\bar{a}_1(x)$ is zero, then we do not need to reduce $a'_{e_1}(x)$ modulo $M_p(x)$ according to the remark at the end of Section \ref{sec:lu}. The parity bits in the $\delta$ erased parity columns can be recovered by~\eqref{EVENODD:diag}.

Note that column $k$ is assumed to be non-failure, as column $k$ is needed to compute the bits of the augmented array by \eqref{eq:augmented1} and~\eqref{eq:augmented2}.

\subsection{Erasure Decoding of $\textsf{RDP}(p,k,r)$}

Similar to the decoding for   $\textsf{EVENODD}(p,k,r)$, we assume that $\gamma$ information columns indexed by $e_{1},\ldots,e_{\gamma}$ and $\delta$ parity columns $f_{1},\ldots,f_{\delta}$ of $\textsf{RDP}(p,k,r)$ are erased with $0\leq e_{1}<\ldots<e_{\gamma}\leq k-1$ and $k+1\leq f_{1}<\ldots<f_{\delta}\leq k+r-1$, where $k\geq \gamma \geq 0$, $r-1\geq \delta \geq 0$ and $\gamma+\delta=\rho\leq r$. Let $f_0=k-1$ and $f_{\delta+1}=k+r-1$ and assume that there exist $\lambda \in \{0,1,\ldots,\delta \}$ such that $f_{\lambda+1}-f_{\lambda}\geq \gamma+1$. The decoding procedure can be divided into two cases:  $\lambda\geq 1$ and  $\lambda=0$.

(i) $\lambda \geq 1$. First, we formulate $k-\gamma$ surviving information polynomials $b_i(x)$ for $i\in\mathcal{A}$ as
\[
b_i(x):=b_{0,i}+b_{1,i}x+\cdots+b_{p-2,i}x^{p-2},
\]
and $\gamma+1$ parity polynomials as
\[
b_{k}(x):=b_{0,k}+b_{1,k}x+\cdots+b_{p-2,k}x^{p-2},
\]
\[
b_{f_{\lambda}+j}(x):=b_{0,f_{\lambda}+j}+\cdots+b_{p-2,f_{\lambda}+j}x^{p-2}+(\sum_{i=0}^{p-2}b_{i,f_{\lambda}+j})x^{p-1},
\]
where $j=1,2,\ldots,\gamma$.
Then, we compute $\gamma$ syndrome polynomials $\bar{b}_1(x),\bar{b}_2(x),\ldots,\bar{b}_{\gamma}(x)$ by
\[
\bar{b}_h(x)=b_{f_{\lambda}+h}(x)+b_{k}(x)x^{g(k) (f_{\lambda}+h-k)}+\sum_{i\in \mathcal{A}}b_{i}(x)x^{g(i) (f_{\lambda}+h-k)},
\]
for $h=1,2,\ldots,\gamma$. It is easy to check that $\bar{b}_1(1)=\cdots=\bar{b}_{\gamma}(1)$. By the remark at the end of Section \ref{sec:lu}, the erased information polynomials can be computed by first solving the following Vandermonde system of linear equations
\begin{align*}
\begin{bmatrix}
\bar{b}_1(x) &  \cdots & \bar{b}_{\gamma}(x)
\end{bmatrix}=
\begin{bmatrix}
x^{g(e_1)(f_{\lambda}+1-k)}b_{e_1}(x) & \cdots & x^{g(e_\gamma)(f_{\lambda}+1-k)}b_{e_{\gamma}}(x)
\end{bmatrix}
\mathbf{V}_{r\times r}(\mathbf{e}),
\end{align*}
over $R_p$ by Algorithm \ref{alg:Vand}, cyclic-left-shifting the resultant $x^{g(e_i)(f_{\lambda}+1-k)}b_{e_i}(x)$ by $g(e_i)(f_{\lambda}+1-k)$ positions for $i=1,2,\ldots,\gamma$, and then reducing $b_{e_1}(x)$ modulo $M_p(x)$. 

(ii) $\lambda=0$. We have that columns $k,k+1,\ldots,k+\gamma$ are not erased. We can obtain $\gamma$ syndrome polynomials $\bar{b}_1(x),\bar{b}_2(x),\ldots,\bar{b}_{\gamma}(x)$ by
\[
\bar{b}_1(x)=b_{k}(x)+\sum_{i\in \mathcal{A}}b_{i}(x),
\]
and
\[
\bar{b}_h(x)=b_{k+h}(x)+b_{k}(x)x^{g(k)\cdot (h-1)}+\sum_{i\in \mathcal{A}}b_{i}(x)x^{g(i)\cdot (h-1)},
\]
for $h=2,3,\ldots,\gamma$. The erased information polynomials can be computed by solving the following Vandermonde linear system
\begin{align*}
\begin{bmatrix}
\bar{b}_1(x) &  \cdots & \bar{b}_{\gamma}(x)
\end{bmatrix}=
\begin{bmatrix}
b_{e_1}(x) & \cdots & b_{e_{\gamma}}(x)
\end{bmatrix}\cdot
\mathbf{V}_{r\times r}(\mathbf{e}).
\end{align*}
We do not need to reduce $b_{e_1}(x)$ modulo $M_p(x)$ according to the remark at the end of Section \ref{sec:lu}, as the last coefficient of $\bar{b}_1(x)$ is zero. Lastly, we can recover the $\delta$ parity columns by \eqref{RDP:diag} for $i=0,1,\ldots,p-2$ and $\ell =f_{1}-k,\ldots,f_{\delta}-k$.

Note that we need column $k$ to obtain the syndrome polynomials in the decoding procedure, so column $k$ is assumed to be non-failure column.

\textbf{Remark.} In the erasure decoding, the assumption that there exist $\lambda \in \{0,1,\ldots,\delta \}$ such that $f_{\lambda+1}-f_{\lambda}\geq \gamma+1$ is necessary; otherwise, we cannot obtain the Vandermonde linear system and Algorithm \ref{alg:Vand} is not applicable. The traditional decoding method, such as Cramer's rule, can be used to recover the failures if the assumption is not satisfied. In the next section, we consider the decoding complexity for two codes when the assumption is satisfied.


\section{Decoding Complexity}
In this section, we evaluate the decoding complexity for $\textsf{EVENODD}(p,k,r)$ and $\textsf{RDP}(p,k,r)$. We determine the \emph{normalized decoding complexity} as the ratio of the decoding complexity to the number of information bits.

When $r=3$, some specific decoding methods \cite{huang2008star,Jiang2013Improved,wang2012triple,RTP12,Huang2016An} are proposed to optimize the decoding complexity of three information erasures, such as the decoding method for STAR \cite{huang2008star} and the decoding method for Triple-Star \cite{wang2012triple}. However, all those decoding methods \cite{huang2008star,Jiang2013Improved,wang2012triple,RTP12,Huang2016An} only focus on the specific codes with $r=3$ and cannot be generalized for $r\geq 4$. In the following, we evaluate the decoding complexity for more than three information erasures.


\begin{theorem}
Suppose that $\gamma$ information columns and $\delta$ parity columns $f_{1},\ldots,f_{\delta}$ are erased. Let $f_0=k-1$ and $f_{\delta+1}=k+r-1$, we assume that there exist $\lambda \in \{0,1,\ldots,\delta \}$ such that $f_{\lambda+1}-f_{\lambda}\geq \gamma+1$. We employ Algorithm \ref{alg:Vand} to recover the $\gamma$ information erasures and recover the failure parity columns by re-encoding the parity bits. The decoding complexity of $\textsf{EVENODD}(p,k,r)$ is
\begin{equation}
p(\gamma k+\frac{3\gamma^2}{4}-\frac{\gamma}{4}+\frac{5}{2})-\gamma k-\frac{\gamma^2}{4}-\frac{5\gamma}{4}-\frac{5}{2}
+\delta(kp-k-1) \text{ when } \lambda>0,
\label{eq:de}
\end{equation}
\begin{equation}
p(\gamma k+\frac{3\gamma^2}{4}-\frac{\gamma}{4}-\frac{1}{2})-\gamma k-\frac{\gamma^2}{4}-\frac{5\gamma}{4}+\frac{1}{2}
+\delta(kp-k-1) \text{ when } \lambda=0.
\label{eq:de2}
\end{equation}
The decoding complexity of $\textsf{RDP}(p,k,r)$ is
\begin{equation}
p(\gamma k+\frac{3\gamma^2}{4}-\frac{\gamma}{4}+\frac{3}{2})-\gamma k-\frac{\gamma^2}{4}-\frac{9\gamma}{4}-\frac{1}{2}+\delta k(p-2) \text{ when } \lambda >0,
\label{eq:de1}
\end{equation}
\begin{equation}
p(\gamma k+\frac{3\gamma^2}{4}-\frac{\gamma}{4}-\frac{3}{2})-\gamma k-\frac{\gamma^2}{4}-\frac{9\gamma}{4}+\frac{7}{2}+\delta k(p-2) \text{ when } \lambda=0.
\label{eq:de3}
\end{equation}
\label{lm:com}
\end{theorem}
\begin{proof}
Consider the decoding process of $\textsf{EVENODD}(p,k,r)$. When $\lambda>0$, we compute the bits $a'_{i,k+\ell}$ of the augmented array from $\textsf{EVENODD}(p,k,r)$ by \eqref{eq:augmented1} and \eqref{eq:augmented2} for $\ell=f_{\lambda}+1-k,\ldots,f_{\lambda}+\gamma-k$. We first compute $\sum_{i=0}^{p-2}a_{i,k}$, and then compute $a'_{p-1,k+\ell}$ by \eqref{eq:augmented1} and $a'_{i,k+\ell}$ by \eqref{eq:augmented2}. 
Thus, the total number of XORs involved in computing the bits $a'_{i,k+\ell}$ is $2(p-1)\gamma+(p-2)$. We now obtain $k-\gamma$ information polynomials in \eqref{eq:infor-poly} and $\gamma$ parity polynomials in~\eqref{eq:parity-poly}. Next, we subtract $k-\gamma$ surviving information polynomials from the $\gamma$ parity polynomials to obtain $\gamma$ syndrome polynomials by~\eqref{eq:syndrome} that takes $\gamma(k-\gamma)(p-1)$ XORs. 
The $\gamma$ information polynomials are obtained by solving the Vandermonde system of equations by using Algorithm \ref{alg:Vand}, of which the computational complexity is 
\[
\gamma(\gamma-1)p+(\gamma-1)(p-3)+(\gamma-1)(\gamma-2)(3p-5)/4
\] 
according to Theorem \ref{thm:lucom}.
Since the last $\gamma-1$ output polynomials of Algorithm \ref{alg:Vand} are exactly the last $\gamma-1$ information polynomials of $\textsf{EVENODD}(p,k,r)$, we only need to reduce the first polynomial modulo $M_p(x)$, which takes at most $p-1$ XORs. The erased $\delta$ parity columns can be recovered by \eqref{EVENODD:diag} and the complexity is $\delta(kp-k-1)$. Therefore, the decoding complexity of $\textsf{EVENODD}(p,k,r)$ is \eqref{eq:de} when $\lambda>0$.

When $\lambda=0$, there are two differences compared with the case of $\lambda>0$. First, we only need to compute the bits of the augmented array for $\gamma-1$ parity columns and the complexity is $2(p-1)(\gamma-1)+(p-2)$, as the bits in the first parity column of the augmented array are the same as those of the first parity column of $\textsf{EVENODD}(p,k,r)$, Second, we do not need to reduce the first polynomial modulo $M_p(x)$ after solving the Vandermonde system. Therefore, the decoding complexity of $\lambda=0$ has $3p-3$ XORs reduction and results in \eqref{eq:de2}.

In $\textsf{RDP}(p,k,r)$, computing $r$ syndrome polynomials takes 
$$\gamma(p-2)+\gamma(k-\gamma+1)(p-1)$$ XORs when $\lambda\geq 1$ and
$$(\gamma-1)(p-2)+(k-\gamma)(p-1)+(\gamma-1)(k-\gamma+1)(p-1)$$ XORs when $\lambda= 0$. Similar to $\textsf{EVENODD}(p,k,r)$, the Vandermonde linear system can be solved by Algorithm \ref{alg:Vand} with complexity 
\[
\gamma(\gamma-1)p+(\gamma-1)(p-3)+(\gamma-1)(\gamma-2)(3p-5)/4
\] 
XORs. Reducing one polynomial modulo $M_p(x)$ takes at most $p-1$ XORs  when $\lambda>0$.  The $\delta$ parity columns are recovered by \eqref{RDP:diag} and its complexity is $\delta k(p-2)$. Therefore, the total number of XORs involved in the decoding process results in \eqref{eq:de1} for $\lambda\geq 1$ and \eqref{eq:de3} for $\lambda=0$.
\end{proof}

Blaum-Roth decoding method~\cite{BlaumRoth93} proposed for decoding Blaum-Roth codes is also applicable to the decoding of $\textsf{EVENODD}(p,k,r)$. Suppose that $\gamma$ information columns and $\delta$ parity columns $f_{1},\ldots,f_{\delta}$ are erased with the assumption that there exist $\lambda \in \{0,1,\ldots,\delta \}$ such that $f_{\lambda+1}-f_{\lambda}\geq \gamma+1$. If one employs the Blaum-Roth decoding method to recover the information erasures and recover the failure parity columns by re-encoding the parity bits, the decoding complexity of $\textsf{EVENODD}(p,k,r)$ is~\cite{subedi2013comprehensive}
\[
\gamma(k+\gamma)p+(3\gamma^2+0.5\gamma)p+\gamma^2-0.5\gamma+\delta(kp-k-1).
\]
 The Blaum-Roth decoding method cannot be directly employed on the erasure decoding for $\textsf{RDP}(p,k,r)$. However, one can first transform $\lambda$ parity columns of $\textsf{RDP}(p,k,r)$ into the form of $\textsf{EVENODD}(p,k,r)$ and then recover the erased information columns by the decoding method of $\textsf{EVENODD}(p,k,r)$. Let $a_{i,j}=b_{i,j}$ for $i=0,1,\ldots,p-2$ and $j=0,1,\ldots,k-1$. That is, the information bits of $\textsf{RDP}(p,k,r)$ and $\textsf{EVENODD}(p,k,r)$ are the same. We then have $a_{i,k}=b_{i,k}$ by \eqref{EVENODD:row} and \eqref{RDP:row} and  
\begin{align}
&\sum_{i=0}^{p-2}b_{i,k+\ell}+b_{p-1-\ell g(k),k}
=\sum_{i=0}^{p-2}\sum_{j=0}^{k}b_{i-\ell g(j),j}+b_{p-1-\ell g(k),k}\label{eq:BR-1}\\
=&\sum_{i=0}^{p-1}\sum_{j=0}^{k}b_{i-\ell g(j),j}+\sum_{j=0}^{k}b_{p-1-\ell g(j),j}+b_{p-1-\ell g(k),k}\nonumber\\
=&\sum_{j=0}^{k}b_{p-1-\ell g(j),j}+b_{p-1-\ell g(k),k}\label{eq:BR-3}\\
=&\sum_{j=0}^{k-1}b_{p-1-\ell g(j),j}=\sum_{j=0}^{k-1}a_{p-1-\ell g(j),j}=a_{p-1,k+\ell},\nonumber
\end{align}
where \eqref{eq:BR-1} comes from \eqref{RDP:diag}, \eqref{eq:BR-3} comes from \eqref{RDP:row} and
\[
\{-\ell g(j),1-\ell g(j), \cdots,p-1-\ell g(j)\}=\{0,1,\cdots,p-1\} \bmod p.
\]
Therefore,  when $\lambda>0$, we can transform $\lambda$ parity columns of $\textsf{RDP}(p,k,r)$ into the form of $\textsf{EVENODD}(p,k,r)$ by
\[
a_{p-1,k+\ell}=\sum_{i=0}^{p-2}b_{i,k+\ell}+b_{p-1-\ell g(k),k}
\]
and 
\[
a_{i,k+\ell}=b_{i,k+\ell}+b_{i-\ell g(k),k}+a_{p-1,k+\ell}
\]
for $\ell=f_{\lambda}+1-k,\ldots,f_{\lambda}+\gamma-k$ and $i=0,1,\ldots,p-2$. When $\lambda=0$, we only need to transform the bits for $\ell=1,\ldots,\gamma-1$ and $i=0,1,\ldots,p-2$, as column $k$ of $\textsf{EVENODD}(p,k,r)$ is the same as column $k$ of $\textsf{RDP}(p,k,r)$. Then we employ the Blaum-Roth decoding method to obtain the erased $\gamma$ information columns of $\textsf{EVENODD}(p,k,r)$. Lastly, we recover $\delta$ parity columns by \eqref{RDP:diag}. The decoding complexity is then
\[
\gamma(k+\gamma)p+(3\gamma^2+3.5\gamma)p+\gamma^2-0.5\gamma+\delta(kp-2k)-3 \text{ for } \lambda>0,
\]
\[
\gamma(k+\gamma)p+(3\gamma^2+3.5\gamma-3)p+\gamma^2-3.5\gamma+\delta(kp-2k)+3 \text{ for } \lambda=0. 
\]

Note that we can recover the erased parity columns by encoding the parity bits according to the definition for both $\textsf{EVENODD}(p,k,r)$ and $\textsf{RDP}(p,k,r)$ after recovering all the erased information bits. Therefore, the main difference of the decoding complexity between the proposed LU decoding method and the Blaum-Roth decoding method lies in the complexity of decoding the Vandermonde linear system, i.e., the erasure decoding of information failures. 
In the following, we consider a special case where $\delta=0$. We evaluate the decoding complexity of $\gamma$  information erasures for the proposed  LU decoding method and the Blaum-Roth decoding method.

\begin{figure}[tbh]
\subfigure[$r=4$ and $p$ ranges from 5 to 59.]{\centering{}\includegraphics[width=0.45\textwidth]{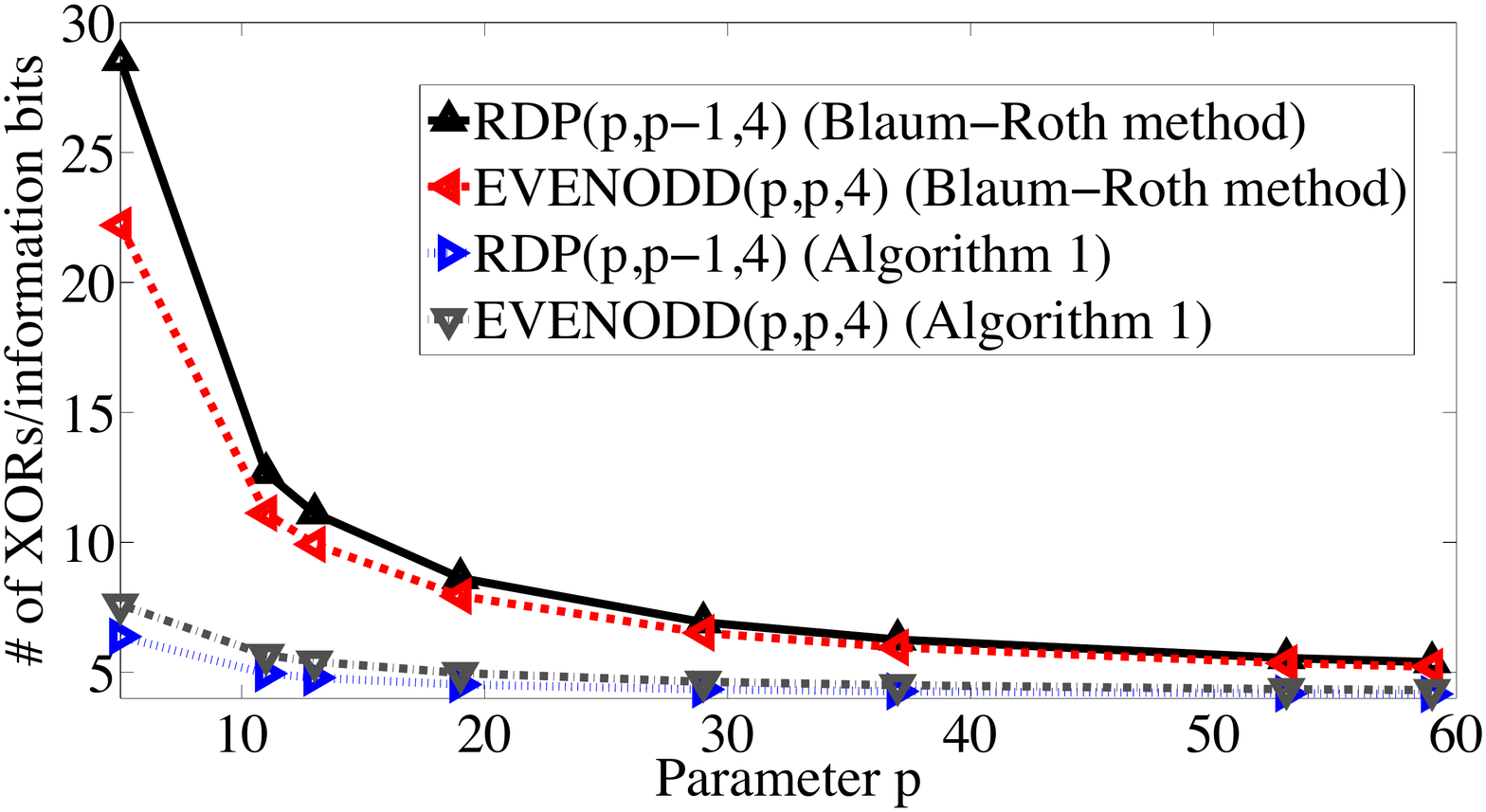}}
\subfigure[$r=5$ and $p$ ranges from 5 to 59.]{\centering{}\includegraphics[width=0.45\textwidth]{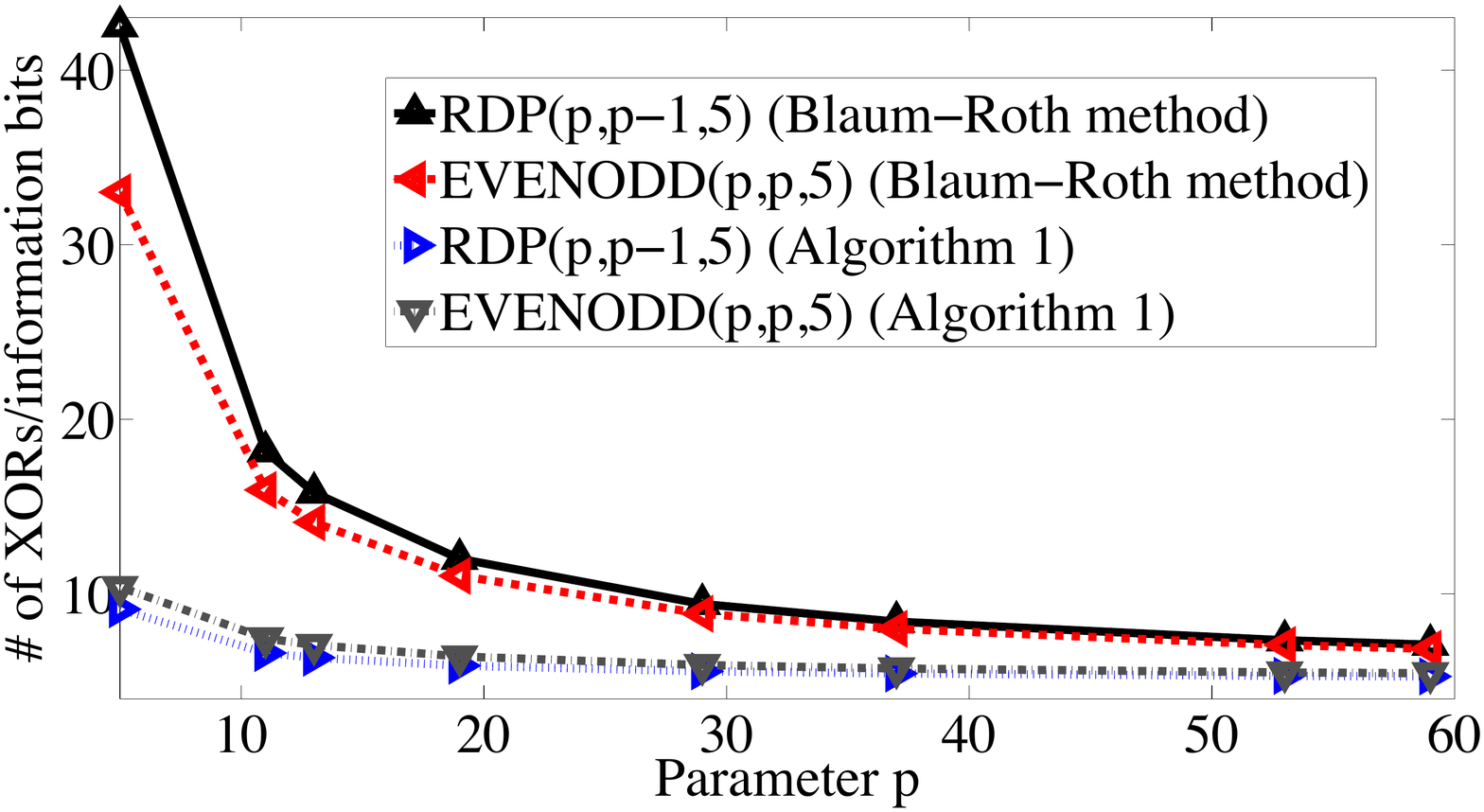}}
\caption{The normalized decoding complexity of $\gamma=r$ information erasures $\textsf{EVENODD}(p,p,r)$ and $\textsf{RDP}(p,p-1,r)$ by Algorithm \ref{alg:Vand} and by Blaum-Roth decoding method for $r=4,5$.}
\vspace{-0.5cm}
\label{fig:de}
\end{figure}

For fair comparison, we let $k=p$ for $\textsf{EVENODD}(p,k,r)$ and $k=p-1$ for $\textsf{RDP}(p,k,r)$.  According to Lemma \ref{lm:com}, the decoding complexity of $\gamma$ information erasures of $\textsf{EVENODD}(p,p,r)$ and $\textsf{RDP}(p,p-1,r)$ by Algorithm \ref{alg:Vand} is
\[
p(\gamma p+\frac{3\gamma^2}{4}-\frac{5\gamma}{4}-\frac{1}{2})-\frac{\gamma^2}{4}-\frac{5\gamma}{4}+\frac{1}{2}, \text{ and }
\]
\[
p(\gamma (p-1)+\frac{3\gamma^2}{4}-\frac{5\gamma}{4}-\frac{3}{2})-\frac{\gamma^2}{4}-\frac{5\gamma}{4}+\frac{7}{2},
\]
respectively. 

When $r=4$ and $5$, the normalized decoding complexity of $\gamma=r$ information erasures of $\textsf{EVENODD}(p,p,r)$ and $\textsf{RDP}(p,p-1,r)$ by Algorithm \ref{alg:Vand} and by Blaum-Roth decoding method is shown in Fig. \ref{fig:de}. One can observe that $\textsf{EVENODD}(p,p,r)$ and $\textsf{RDP}(p,p-1,r)$ decoded by LU decoding method is more efficient than by the Blaum-Roth decoding method. When $r=4$ and $p$ ranges from $5$ to $59$, the decoding complexity of $\textsf{EVENODD}(p,p,4)$ and $\textsf{RDP}(p,p-1,4)$ by Algorithm \ref{alg:Vand} has 20.1\% to 71.3\% and 22.7\% to 77.7\% reduction on that by the Blaum-Roth decoding method, respectively. When $r=5$, the complexity reduction is 20.4\% to 68.5\% and 25.7\% to 78.5\% for $\textsf{EVENODD}(p,p,5)$ and $\textsf{RDP}(p,p-1,5)$, respectively. The reduction increases when $p$ is small and $r$ is large. For example, $\textsf{RDP}(p,p-1,r)$ decoded by Algorithm \ref{alg:Vand} has 78.5\% less decoding complexity than that by the Blaum-Roth decoding method when $p=5$ and $r=5$. 



The reasons that the decoding complexity of $\textsf{EVENODD}(p,p,r)$ and $\textsf{RDP}(p,p-1,r)$ by the LU decoding method is less than that by the Blaum-Roth decoding method are summarized as follows. First, the Blaum-Roth decoding method is operated over the ring $\mathbb{F}_2[x]/M_p(x)$; however, we show  that the Vandermonde linear systems over $\mathbb{F}_2[x]/M_p(x)$ can be computed by first solving the Vandermonde linear systems over $R_p$ and then reducing the results by $M_p(x)$ modulus.\footnote{When there are only information erasures,  modulo $M_p(x)$ is not needed in the decoding procedure.} The operation of multiplication and division over $R_p$ is more efficient than that over $\mathbb{F}_2[x]/M_p(x)$. Second, the proposed LU decoding method is more efficient than the Blaum-Roth decoding method.


\section{Discussion and Conclusions}
\label{sec:discussions}


In this paper, we present a unified construction of EVENODD  codes and RDP  codes, which can be viewed as a generalization of extended EVENODD codes and generalized RDP codes. 
 Moreover, an efficient LU decoding method is proposed for EVENODD codes and RDP  codes, and we show that the LU decoding method requires less XOR operations than the existing algorithm when more than three information columns are failure.

In most existing Vandermonde array codes, the parity bits are computed along some straight lines in the array, while the parity bits of the proposed EVENODD codes and RDP  codes are computed along some polygonal lines in the array. By this generalization, EVENODD  codes and RDP codes may have more design space for decoding algorithm when there is a failure column. For example, assume that the first column of $\textsf{EVENODD}(5,3,3;(0,1,4))$ in Table \ref{table:evenodd} is erased, we want to recover the erased column by downloading some bits from other four surviving columns. We can recover bits $a_{0,0},a_{2,0}$ by
\begin{align*}
a_{0,0}=a_{0,1}+a_{0,2}+a_{0,3},
a_{2,0}=a_{2,1}+a_{2,2}+a_{2,3},
\end{align*}
and  bits $a_{1,0},a_{3,0}$ by
\begin{align*}
a_{1,0}=a_{0,1}+a_{2,2}+a_{4,4}+a_{1,4},
a_{3,0}=a_{2,1}+a_{4,4}+a_{3,4},
\end{align*}
where $a_{4,4}$ can be computed as $a_{4,4}=a_{3,1}+a_{0,2}$. In total, $9$ bits are downloaded to recover the first column. For original EVENODD codes, an erased information column is covered by downloading at least $0.75(p-1)$ bits from each of the helped $k+1$ columns \cite{wang2010rebuilding}. Hence, the total number of bits  to be downloaded to recover the first column of original EVENODD codes is at least $12$. Therefore, one may design a decoding algorithm for an information failure such that the number of bits downloaded is less than that of the original EVENODD codes.  Designing an algorithm to recover a failure column for general parameters is then an interesting future work.

\appendices

\bibliographystyle{IEEEtran}

\end{document}